\documentclass[journal]{IEEEtran}

\ifCLASSINFOpdf
\else
\fi

\pdfoutput=1
\usepackage[cmex10]{amsmath}

\usepackage{enumitem}
\usepackage{graphicx}
\usepackage{color}
\usepackage{amsmath}
\usepackage{mathtools}
\usepackage{multicol}
\usepackage{multirow}
\usepackage[english]{babel}
\usepackage{blindtext}
\usepackage{algorithm}
\usepackage{algorithmic}
\usepackage{balance}
\usepackage{amsfonts}
\usepackage{bm}
\usepackage{stfloats}
\usepackage{subfig}
\usepackage{amsthm}
\usepackage{amssymb}
\usepackage{setspace}
\usepackage[nosort]{cite}
\usepackage{CJK}
\usepackage{cite}
\usepackage{caption}
\usepackage{comment}
\usepackage{array,multirow}
\usepackage{graphicx}
\newtheorem{lemma}{Lemma}
\newtheorem{theorem}{Theorem}

\newtheorem{proposition}{Proposition}
\theoremstyle{plain}

\usepackage[table]{xcolor}
\usepackage{amsmath}
\usepackage{graphicx}
\usepackage{float}
\usepackage{cleveref}
\usepackage{soul} 
\usepackage{color, xcolor} 
\allowdisplaybreaks[4]

\begin{document}

\title{Enriched K-Tier Heterogeneous Satellite Networks Model with User Association Policies}

\author{Zhuhang Li, \IEEEauthorblockN{Bodong Shang,~\IEEEmembership{Member,~IEEE}}
\thanks{\textit{(Corresponding author: Bodong Shang)}}
\thanks{Z. Li and B. Shang are with Zhejiang Key Laboratory of Industrial Intelligence and Digital Twin, Eastern Institute of Technology, Ningbo, China, and also with State Key Laboratory of Integrated Services Networks, Xidian University, Xi’an, China (E-mails: jdhdyxj@163.com, bdshang@eitech.edu.cn).}}




\maketitle

\begin{abstract}
In the rapid evolution of the non-terrestrial networks (NTNs), satellite communication has emerged as a focal area of research 
due to its critical role in enabling seamless global connectivity. In this paper, we investigate two representative user association policies (UAPs) for multi-tier heterogeneous satellite networks (HetSatNets), namely the nearest satellite UAP and the maximum signal-to-interference-plus-noise-ratio (max-SINR) satellite UAP, where each tier is characterized by a distinct constellation configuration and transmission pattern. Employing stochastic geometric, we analyze various intermediate system aspects, including the probability of a typical user accessing each satellite tier, the aggregated interference power, and their corresponding Laplace transforms (LTs) under both UAPs. Subsequently, we derive explicit expressions for coverage probability (CP), non-handover probability (NHP), and time delay outage probability (DOP) of the typical user. Furthermore, we propose a novel weighted metric (WM) that integrates CP, NHP, and DOP to explore their trade-offs in the system design. The robustness of the theoretical framework is verified is verified through Monte Carlo simulations 
calibrated with the actual Starlink constellation, affirming the precision of our analytical approach. The empirical findings underscore an optimal UAP in various HetSatNet scenarios regarding CP, NHP, and DOP.. 
\end{abstract}


%
\IEEEpeerreviewmaketitle

\section{Introduction}

\IEEEPARstart{I}{n} recent years, there has been widespread recognition of the rapid growth in traditional terrestrial wireless communications, marked by an increase in both user numbers and the diversity of supported services  \cite{ohlen2016data,nguyen20216g}. However, limitations in network capacity and coverage hinder terrestrial systems from providing universally reliable, high-data-rate wireless services, particularly in challenging environments like oceans and mountains \cite{kawamoto2013traffic,9520380}. Non-terrestrial networks (NTN), especially Low Earth Orbit (LEO) satellites, have emerged as a promising solution for achieving global seamless connectivity \cite{wei2021hybrid,li2020maritime,xu2023space,10685064}, offering access to remote, rural, oceanic, and mountainous areas. 
With the increasing exploration of deep space and the development of extensive satellite constellations, deploying satellites at various orbital altitudes has become a growing trend. 

To address these connectivity challenges, the deployment of satellite networks capable of complementing terrestrial infrastructure has become essential. In response, several major global initiatives have been launched with the aim of bridging digital divides and enabling worldwide broadband access through advanced satellite constellations.
Various initiatives, notably Kuiper, LeoSat, OneWeb, Starlink, and Telesat, are deploying large-scale LEO satellite networks with multiple satellite altitudes, aiming to establish a comprehensive global communications network \cite{del2019technical}. Particularly, the Starlink project has illustrated the evolving trend towards multi-tier heterogeneous satellite networks (HetSatNet), which are poised to become a cornerstone in the future development of satellite communications \cite{mcdowell2020low}.

Prior research has offered valuable insights into satellite communication networks, particularly using stochastic geometry \cite{10689625, okati2020downlink,al2021analytic,10387244,10430115}. However, most studies focus primarily on single-tier satellite systems with user association based on proximity to the nearest satellite, leaving alternative association policies and multi-tier network analysis not fully elucidated.
For instance, \cite{guo2024user} investigates Distance-based Association (DbA) and Power-based Association (PbA) schemes. However, it overlooks the influence of practical satellite antenna patterns on mutual interference and overall system performance. Furthermore, it fails to consider instantaneous received signal-to-interference-plus-noise ratio (SINR) as a potential association metric, which limits its applicability in real-world scenarios. The study also lacks a thorough comparison of DbA and PbA performance across different satellite densities, altitudes, and transmission powers, leaving the question of optimal association scheme selection in multi-tier HetSatNet with different scenarios unresolved. Furthermore, the issue of increased interference power caused by the increased number of interference satellites remains unaddressed.

In this paper, we model and analyze a multi-tier HetSatNet from a system-level perspective by considering different user association policies (UAPs).
Specifically, the first UAP is called the nearest satellite UAP, where users connect to the closest satellite in HetSatNet with a two-sector antenna beam gain.
The second UAP is called the maximum signal-to-interference-plus-noise-ratio (max-SINR) UAP, where users connect to the satellite offering maximum instantaneous SINR. 
Additionally, we derive interference characteristics within the multi-tier HetSatNet under both UAPs.
Our study examines the impact of varying satellite densities, orbital altitudes, and transmission powers across different tiers under each UAP. 
Given the high relative velocity of satellites, frequent handovers between satellites and the ground user are inevitable. Hence, we also investigate the non-handover probability (NHP) under both UAPs within the multi-tier HetSatNet framework. 
Moreover, recognizing the significance of propagation delay in satellite-to-ground communications, we note that although the max-SINR UAP can provide a higher coverage probability (CP), it concurrently results in increased propagation delay. We propose time delay outage probability (DOP), highlighting the inherent balance between CP and DOP. This trade-off suggests that different UAPs should be employed to optimize performance in accordance with various HetSatNet deployments and the specific requirements of diverse communication scenarios. In addition to the traditional performance metrics such as CP, NHP, and DOP, a Weighted Metric (WM) framework is proposed which integrates these three metrics through adjustable weights. This design enables the network to adaptively balance coverage, stability, and latency according to different application requirements, providing a comprehensive performance evaluation criterion for multi-tier satellite networks.

\subsection{Related Work}
Within terrestrial cellular networks, deploying diverse infrastructure elements, such as micro, pico, and femtocells, along with distributed antennas, has led to the rise of $K$-tier Heterogeneous Cellular Networks (HCNs).
\cite{andrews2016primer, dhillon2012modeling,dhillon2011coverage,madhusudhanan2016analysis, 10530195} have proposed a model elucidating the characteristics of HCNs, wherein $K$ tiers of base stations (BSs) were situated in terrestrial networks, introducing disparities in average transmit power, data rate, and BS density across tiers. Authors in \cite{andrews2016primer} have examined average power-based cell association and instantaneous power-based cell selection for $K$-tier HCNs. \cite{dhillon2012modeling} has formulated expressions delineating CP and average rate across the entire network, under the assumption that a mobile user connects to the strongest candidate BS. Although the above works have provided average power-based association and instantaneous power-based selection policies for $K$-tier networks, they focused on terrestrial networks, lacking the fundamentals of spherical geometry and interference from space. Modeling and analyzing multi-tier HetSatNet with various UAPs is crucial for developing NTN.

Stochastic geometry serves as a foundational mathematical framework for characterizing the spatial distribution of wireless network nodes \cite{9777886,9178984}. \cite{park2022tractable, okati2022nonhomogeneous,jung2022performance,wang2021multi} have employed stochastic geometric methodologies to address the intricacies of satellite-to-ground communication. In \cite{park2022tractable}, the spatial arrangements of satellites and users are modeled through the Poisson point processes (PPPs). The study has provided analytical formulations to assess the CP of satellites, establishing a robust lower bound, and has determined the optimal number of satellites for achieving optimal coverage. \cite{okati2022nonhomogeneous} has modeled the LEO network as a nonhomogeneous PPP, with the intensity being a function of satellite distribution in terms of size, altitude, and orbital plane inclination. However, these studies focus on single-tier satellite networks. 

For multi-tier satellite networks, \cite{choi2024cox} has modeled satellites situated on orbits varied in altitude as linear PPPs, forming a Cox point process. \cite{ talgat2020stochastic } derived the CP for the described setup, assuming that LEO satellites are placed at different altitudes and the locations of the LEO satellites are modeled as a binomial point process (BPP).
\cite{ qiu2023interference } used stochastic geometry to simulate the interference of multi-tier NGSO networks by modeling the locations of satellites as a randomly distributed points process in a 3-D space.
\cite{wang2021multi} has considered a 2-tier satellite-terrestrial communication system, and its predominant emphasis remains on analyzing backhaul capacity. \cite{yim2024modeling} analyzed the downlink coverage performance of multi-tier integrated satellite-terrestrial networks. Although the above works explored multi-tier satellite networks, the user-to-satellite association strategies, the antenna gains, and the variation of powers of satellites over different tiers are not considered.
\cite{guo2024user} proposed three association schemes for the multi-tier LEO satellite-based networks, i.e., the DbA scheme, the PbA scheme, and a random selection. The authors mentioned satellites in the same tier have identical transmit power and antenna array gain, and thus, the nearest one to the UE provides the strongest average received signal. Hence, the set of candidate LEO satellites is identical to that formulated for the DbA scheme. Therefore, the difference between DbA and PbA is not apparent, and the antenna gain variation and max-SINR UAP are not considered and compared.

In contrast to previous work, we present a multi-tier HetSatNet architecture utilizing the nearest satellite UAP and a max-SINR satellite UAP. Additionally, we incorporate a two-sector antenna beam pattern, which is more realistic in real-world satellite communication systems. We further investigate the CP and NHP of a typical user under the above two UAPs within the multi-tier satellite system. Additionally, we introduce a novel weighted metric comprising CP, NHP, and DOP in a balanced manner, thereby offering a comprehensive system performance assessment.

\subsection{Contributions}
In this paper, we analyze a multi-tier downlink satellite network with two UAPs, i.e., 
the nearest satellite UAP and the max-SINR satellite UAP. Considering the high mobility of satellites, which results in frequent handovers and increased propagation delays in satellite-ground communication, we introduce two key performance metrics: the NHP and the DOP. These metrics facilitate the development of a novel weighted metric that integrates CP, NHP, and DOP, offering valuable design insights for diverse communication scenarios. Moreover, a two-sector antenna gain model is employed to characterize both the desired signal and interference in 
downlink communications, yielding a more realistic representation of the system performance.
\begin{itemize}
    \item We use stochastic geometry to model both the nearest and max-SINR satellite UAPs, which 
    provides a rigorous theoretical foundation for understanding the CP under the max-SINR UAP and offers new insights into the mathematical relationship between satellite tiers, interference, and system performance, 
underscoring the uniqueness of the max-SINR UAP.
\end{itemize}
\begin{itemize}
    \item We introduce two other key metrics, i.e., NHP and DOP. The NHP quantifies the likelihood that the handover time of the serving satellite remains below a predefined threshold, thereby reflecting the stability of satellite-user connections.
The DOP measures the probability that the satellite-user propagation delay falls within an acceptable limit, 
highlighting latency-sensitive performance.
These metrics are critical in understanding the impact of satellite mobility and propagation delays on real-world network performance. 
By adjusting the weights of CP, NHP, and DOP, we observe that the max-SINR UAP performs better when the CP and NHP are prioritized, 
whereas the nearest satellite UAP tends to outperform the max-SINR UAP as the weight of DOP increases. This trade-off is also dependent on satellite deployment density, 
offering valuable insights for policy selection based on application-specific requirements.
\end{itemize}
\begin{itemize}
    \item Through Monte Carlo simulations based on realistic Starlink constellations and theoretical analysis, we investigate the impact of the number of satellite tiers, tier-specific transmit powers, and satellite densities on CP in a multi-tier HetSatNet with two UAPs. Results indicate that the nearest satellite UAP generally performs better in single-tier configurations, while in multi-tier HetSatNets, the max-SINR UAP yields higher CP when higher-tier satellites have stronger transmit power. Additionally, for systems with the same number of tiers, the max-SINR UAP generally outperforms the nearest satellite UAP in terms of CP. The NHP under the max-SINR UAP shows slower degradation over time, suggesting improved connection stability and communication reliability. 
\end{itemize}
\begin{itemize}
    \item We further discuss the practical implications of these findings for real-world satellite network design. 
Our analysis suggests that in configurations involving higher-tier satellites, 
the max-SINR UAP can significantly enhance overall system performance and reliability, 
making it a compelling choice for future satellite constellations.
\end{itemize}

The paper is organized as follows: Section II discusses the network models, channel models, and performance metrics for multi-tier satellite communication systems with two UAPs. Section III and Section IV analyze association probabilities, CPs, and NHPs for systems with two UAPs, respectively. Section V presents simulation results, and Section VI elucidates the research outcomes in the concluding analysis.

\section{System Model}
This section presents the proposed network model, UAPs, satellite-terrestrial channel models, and performance metrics for analyzing multi-tier satellite communication networks.
The main variables utilized in this paper are delineated in Table I. 
\subsection{Network Model}
\captionsetup{font={scriptsize}}
\begin{figure}
\begin{center}
\setlength{\abovecaptionskip}{+0.2cm}
\setlength{\belowcaptionskip}{-0.5cm}
\centering
  \includegraphics[width=3.4in, height=1.8in]{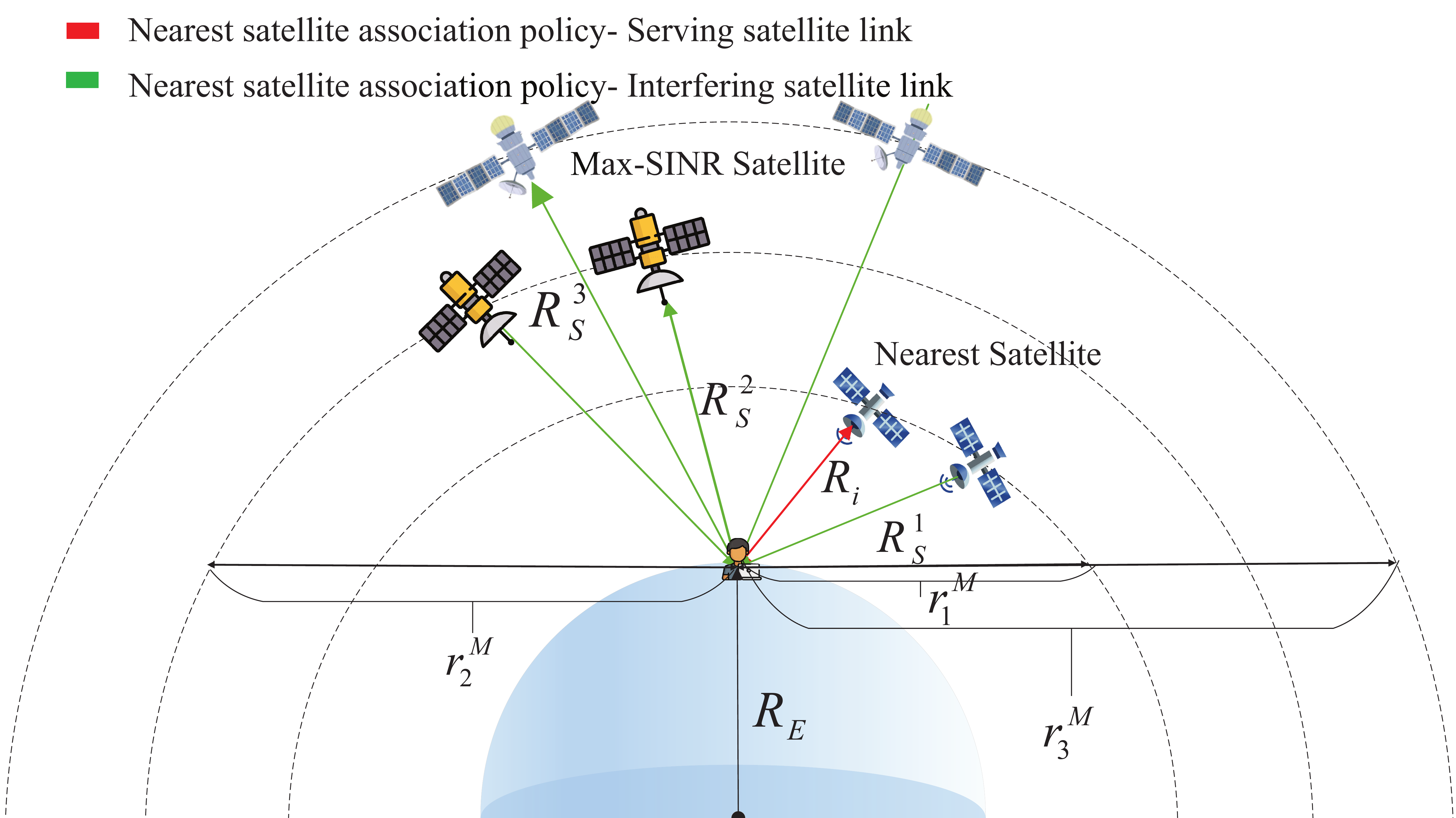}
\renewcommand\figurename{Fig.}
\caption{\scriptsize Multi-tier satellite-terrestrial communication system with nearest satellite UAP.}
\label{fig: Fig1}
\end{center}
\end{figure}

\captionsetup{font={scriptsize}}
\begin{figure}
\begin{center}
\setlength{\abovecaptionskip}{+0.2cm}
\setlength{\belowcaptionskip}{-0.5cm}
\centering
  \includegraphics[width=3.4in, height=1.8in]{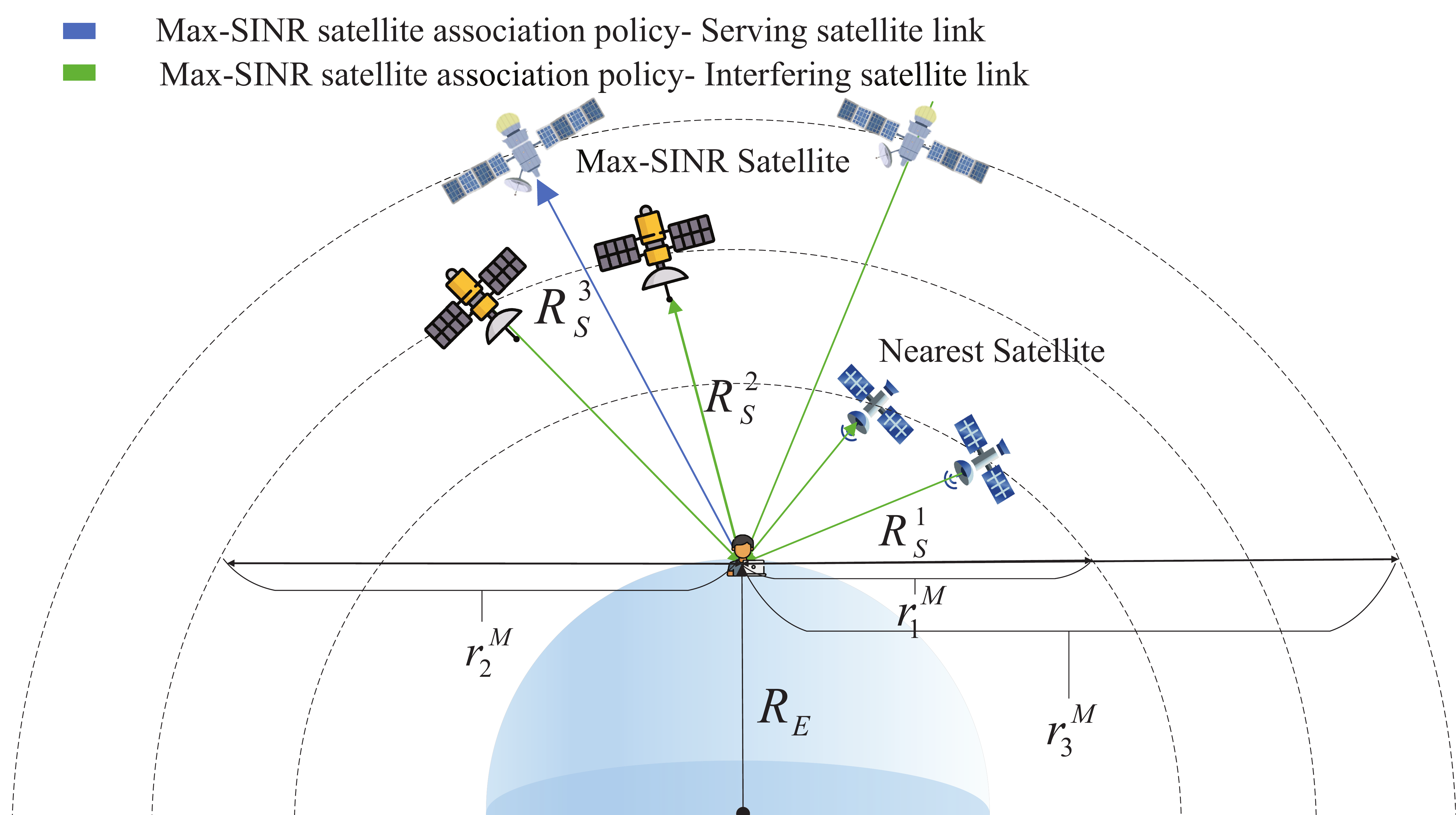
  }
\renewcommand\figurename{Fig.}
\caption{\scriptsize Multi-tier satellite-terrestrial communication system with max-SINR satellite UAP.}
\label{fig: Fig2}
\end{center}
\end{figure}
This study explores the performance of a multi-tier communication network that connects satellite and terrestrial users, as illustrated in Fig. \ref{fig: Fig1} and Fig. \ref{fig: Fig2}. The evaluation is performed using the nearest satellite and max-SINR satellite UAPs, respectively. We consider a multi-tier HetSatNet, with satellites distributed on multiple spherical surfaces of radius $R_S^i = R_E + h_i$, where $h_i$ is the orbital height of tier $i$ satellites. The satellite placement in each tier follows a Homogeneous Spatial Poisson Point Process (HSPPP) with a density of ${\lambda _{S}^{i}}$, represented as ${\Phi_{S}^{i}} = \left\{ {{x_{1}^{i}},{x_{2}^{i}},...{x_{N}^{i}}} \right\}$, where $N$ in tier $i$ is a Poisson-distributed variable with a mean of ${\lambda _{S}^{i}}4\pi \left({R_{S}^{i}}\right)^2$. Similarly, users are uniformly distributed over the Earth's surface with a radius of ${R_E}$, modeled by an HSPPP with density ${\lambda _U}$, denoted as ${\Phi _U} = \left\{ {{u_1},{u_2},...{u_F}} \right\}$. The user number $F$ follows a Poisson distribution with its mean calculated as ${\lambda _U}4\pi R_E^2$.

For downlink communications between satellites and terrestrial users, and for the sake of generality, we analyze the downlink performance as experienced by a typical user located at $\left( {0,0,{R_E}} \right)$.
The satellite's visible region in the $i$-th tier, ${{\cal A}_{vis}^{i}}$, is geometrically a dome-shaped area within the user's visual field, crucial for signal reception from overhead satellites and a prime zone for signal interference. The satellites reside in ${{\cal A}_{vis}^{i}}$ from set $\Phi_{vis,i}^S$.
The maximum distance at which a satellite in tier $i$ can influence a typical user is defined as ${r_{max}^{i}} = \sqrt {\left({R_{S}^{i}}\right)^2 - R_E^2}$, as depicted in Fig. \ref{fig: Fig1} and Fig. \ref{fig: Fig2}.

Directional beamforming with fixed-beam antennas is employed on the satellites, ensuring that the beams remain oriented towards the subsatellite point. However, for mathematical tractability, several research studies \cite{9678973, alkhateeb2017initial, zhu2018secrecy, dabiri20203d,chen2023coverage} have assumed that satellites use sectorized beam patterns, where the antenna gains of the main and side lobes are denoted as $G_{ml}$ and $G_{sl}$, respectively. Thus, the transmit antenna gain of the satellite is expressed as
\begin{equation}
\label{1}
{G_i} = \left\{ \begin{array}{l}
{G_{ml}^{i}},{\rm{ }}h_i \le {R_{i}} \le {r_m^i},\\
{G_{sl}^{i}}, \enspace {\rm{  otherwise}},
\end{array} \right.
\end{equation}
where $i$ represents $i$-th tier, the distance ${r_m^i}$ denotes the maximum distance at which a satellite in tier $i$ can reach the typical user via its main lobe, which is given in (\ref{3}). $R_i$ represents the distance between the typical user and the serving satellite. We assume the typical user is equipped with an omnidirectional antenna, of which receiving antenna gain is normalized to one. Due to the two-sector antenna beam pattern, the set $\Phi_{vis,i}^S$ is divided into two distinct subsets, i.e., $\Phi _{m}^{i}$, comprising satellites in tier $i$ whose main lobes are directed toward the typical user, and $\Phi _{s}^{i}$, encompassing satellites in tier $i$ whose side lobes are oriented toward the typical user. Spatially, satellites in $\Phi _{m}^{i}$ and $\Phi _{s}^{i}$ are respectively situated on a spherical cap ${\cal A}_{m}^{i}$, and a spherical ring ${\cal A}_{s}^{i}$, designated as the main-lobe reception area and side-lobe reception area, respectively. Defined $\mathcal{A}_r^i$ as the spherical surface that encompasses all tier-$i$ satellites located within a distance $r$ from the typical user. ${\varphi _m^i}$ is the threshold serving angle between the main and side lobes of the beam pattern. 
The contact angle corresponding to half of the ${\varphi _m^i}$ is referred to as the main lobe contact angle and is determined by 
\begin{equation}
\label{2}
\begin{aligned}
{\phi _m^i} = \arcsin \left( {\frac{{{R_S^i}}}{{{R_E}}}\sin \left( \frac{{{\varphi _m^i}} }{2} \right)} \right) - \frac{{\varphi _m^i}}{2}. 
\end{aligned}
\end{equation}
The distance ${r_m^i}$ is the maximum distance that the satellite in tier $i$ towards the typical user with its main lobe, which is calculated by 
\begin{equation}
\label{3}
\begin{aligned}
r_{m}^{i} = \sqrt{\left ( R_S^i \right )  ^2 + R_E^2 - 2R_S^iR_E  cos\left ( \phi _m^i \right )}. 
\end{aligned}
\end{equation}
\subsection{Association Policy}
In the proposed multi-tier HetSatNet communication system model delineated in this paper, we incorporate two distinct UAPs, i.e., the nearest satellite UAP and the max-SINR satellite UAP. 
\subsubsection{Nearest Satellite UAP} Users connect to their nearest satellite within the $K$-tier satellite communication system, as illustrated in Fig. \ref{fig: Fig1}.
\subsubsection{Max-SINR Satellite UAP} Each user connects to the satellite that offers the highest instantaneous SINR, which is shown in Fig. \ref{fig: Fig2}.

\subsection{Channel Model}
The Shadowed Rician (SR) fading model \cite{jung2022performance,8894851,9497773,1198102} is employed to articulate the characteristics of small-scale channel fading, introducing parameters such as the Nakagami fading coefficient $m$, the average power of the scattered component $b_{0}$, and the average power of the line-of-sight component $\Omega$. In alignment with \cite{1198102,9511625,9918046}, a gamma random variable is employed to approximate the PDF of the channel gain $H$ as
\begin{equation}
\label{4}
\begin{aligned}
{f_H}\left( h \right) \approx \frac{1}{{{\beta ^\chi }\Gamma \left( \chi  \right)}}{h^{\chi  - 1}}\exp \left( { - \frac{h}{\beta }} \right),
\end{aligned}
\end{equation} where $\Gamma \left( \chi  \right)$ is the gamma function, $\chi  = \frac{{m{{\left( {2{b_0} + \Omega } \right)}^2}}}{{4mb_0^2 + 4m{b_0}\Omega  + {\Omega ^2}}}$ and $\beta = \frac{{2{b_0} + \Omega }}{\chi }$ are for the shape and scale parameters, respectively.
\subsection{Performance Metrics}
As for the nearest satellite UAP network, the desired signal power is $S_{i}^{E} = {p_i}{G_i}{H_{X_{i}}}{R_{i}}^{ - \alpha }$, where ${p_i}$ and ${G_i}$ are transmit power and satellite antenna gain when accessing tier $i$, separately. ${H_{X_{i}}}$ and $R_{i}$ are the channel power gain and the distance between the serving satellite and the typical user when accessing tier $i$, respectively. The SINR of the typical user is 
\begin{equation}
\label{5}
{\rm{SINR}}_{i}^{E} = \frac{S_{i}^{E}}{{{I^{E}} + {{N^{E}}}}} 
= \frac{{{p_i}{G_i}{H_{X_{i}}}{R_{i}}^{ - \alpha }}}{{{I^{E}} + \sigma^2}},
\end{equation}
where ${I^{E}}$ is the aggregated interference signal power, with further analysis provided in Section III.B. The noise power is denoted as $N^{E}=\sigma^2$. 
The overall CP of the nearest satellite UAP system is given by
\begin{equation}
\label{6}
\mathbb{P}^{E}_{C} = \sum\limits_{i = 1}^K \mathbb{P} {\left\{ {{\rm{SINR}}_i^{E} \ge \gamma _{th}^i|S = i} \right\}\mathbb{P}_{ass,i}^{E}}, 
\end{equation}
where $S=i$ denotes the event of accessing the nearest satellite in tier $i$, $\mathbb{P}_{ass,i}^{E}$ signifies the probability of accessing tier $i$ with the nearest satellite UAP, $\gamma_{th}^{i}$ denotes the SINR threshold in tier $i$.
As for the max-SINR satellite UAP network, the CP of the system is given by
\begin{equation}
\label{7}
\mathbb{P}\left\{ {\mathop {\max }\limits_{i,i \in K} \mathop {\max }\limits_{X \in {\Phi _{vis,i}^S}} {\rm{SINR}}\left( X \right) > {\gamma _{th}^i}} \right\}.
\end{equation}

Inspired by \cite{al2021session}, we propose the NHP, defined as the simultaneous occurrence where the received SINR exceeds a threshold, while the time required for the nearest satellite traversing beyond the dome region surpasses a designated threshold time. Simultaneously, we introduce the concept of time delay outage, which is defined as the event in which the propagation delay between the typical user and the nearest satellite is less than a time delay threshold. We present the NHP and the DOP for both the nearest and max-SINR satellite UAP networks as $\mathbb{P}_{NH}^{E}$, $\mathbb{P}_{NH}^{M}$, $\mathbb{P}_{DO}^{E}$, and $\mathbb{P}_{DO}^{M}$, separately. A detailed
analysis is provided in Section III.C.
Additionally, we introduce a novel weighted metric (WM) composed of the CP, the NHP, and the DOP. The WM for the nearest satellite and max-SINR satellite UAPs can be expressed separately as
\begin{equation}
\label{8}
{{WM}^{E}} = {w_1^N} \mathbb{P}_C^{E} + {w_2^N} \mathbb{P} _{NH}^{E} + {w_3^N} \mathbb{P} _{DO}^{E},
\end{equation}
and
\begin{equation}
\label{9}
{{WM}^{M}} = {w_1^M} \mathbb{P}_C^{M} + {w_2^M} \mathbb{P} _{NH}^{M}+ {w_3^M} \mathbb{P} _{DO}^{M},
\end{equation}
where ${w_1^N}$ and ${w_1^M}$ represent the weighted parameters for CPs under the two UAPs, ${w_2^N}$ and ${w_2^M}$ represent the weighted parameters for NHPs under the same policies, and ${w_3^N}$ and ${w_3^M}$ represent the weighted parameters for DOPs. 

The proposed WM framework is not only an analytical tool but also serves as a practical decision-making mechanism for network design. The introduction of weighting factors $\omega_1$, $\omega_2$, and $\omega_3$ explicitly captures the trade-offs between CP, NHP, and DOP. This enables adaptive adjustment of the UAP strategy in response to diverse operational scenarios. For instance, CP-oriented IoT backhaul applications may prioritize $\omega_1$, while DOP-sensitive scenarios such as telesurgery may assign higher weight to $\omega_3$. The WM-based evaluation thus supports application-specific strategy selection beyond single-metric optimization.
\begin{table*}[t]
\setlength{\abovecaptionskip}{0cm} 
\setlength{\belowcaptionskip}{-0.12cm}
\captionsetup{font={normalsize}}
\caption{Main variables used in this paper}
\centering
\small
\begin{tabular}{l p{7cm} l p{6cm} l p{7cm} l p{7cm}}
\hline
\hline
\textbf{Notation} & \textbf{Description} & \textbf{Notation} & \textbf{Description}  \\
\hline
\hline
${\Phi _S^i}$, ${\Phi _U}$ & Set of satellites in tier $i$; Set of users & $R_E$ & Radius of the Earth (km) \\

${\lambda _{S}^{i}}$, ${\lambda _{U}}$ & Densities of satellites in tier $i$; densities of users & $\alpha$ & Path-loss exponent  \\

$h_i$, $R_S^i$ & Height and radius of satellites in tier $i$ (km) & $ \sigma^2$  & Noise power for the two UAPs \\

$G_{ml}^i, G_{sl}^i$ & Antenna gains of satellites in tier $i$ (dBi) & $\gamma_{th}^{i}$ & SINR threshold of tier $i$ \\

$\varphi_D^i$ & Earth-based contact angle of dome region in tier $i$ (rad) & 
$t_{th}^{i}$ & Threshold time for satellite switching (s) \\

$\chi$ & Nakagami shape parameter&
$c_1$,$c_2$,$c_3$ & Satellite’s
velocities in each tier (m/s) \\

$\mathbb{P}_{NP}^E$, $\mathbb{P}_{NP}^M$ & NHP under the nearest and max-SINR satellite UAPs & $\mathbb{P}_C^E$, $\mathbb{P}_C^M$ & CPs under the nearest and max-SINR satellite UAPs \\

$R_i$ & Distance between the serving satellite in tier $i$ and the typical user (km) & 
$\mathbb{P}_{ass,i}^E$, $\mathbb{P}_{ass,i}^M$ & Probability of accessing
tier $i$ with the nearest satellite and the max-SINR satellite UAP \\

$r_{max}^i$ & Maximum distance between satellite in tier $i$ and the typical user (km) & $I^E$ & Aggregated interference signal power under the nearest satellite UAP\\

$\theta$ & Longitude angle between satellite departure and initial direction (rad) & 
$\omega_1^N$,$\omega_1^M$,$\omega_2^N$,$\omega_2^M$ &  Weighted parameters for CPs and NHPs under both UAPs \\
\hline
\hline
\end{tabular}
\end{table*}
\section{Nearest Satellite Association Policy in Multi-Tier Satellite Networks}
This section is devoted to a comprehensive performance analysis of a multi-tier satellite communication network, specifically focusing on the nearest satellite UAP. Our analytical framework focuses on four essential components, i.e., a rigorous evaluation of association probability, a distance distribution between the serving satellite and the typical user, an examination of the aggregated interference power LT, and a thorough exploration of CP, NHP, and DOP.
These critical elements are integrated within the scope of this section.
\subsection{Association Probability and Distance Distribution}

Within the framework of the nearest satellite UAP, the probability of accessing tier $i$ is based on the condition that the distance between the typical user and the nearest satellite in tier $i$ is less than the corresponding distances from the typical user to the nearest satellites in all other tiers. Based on this framework, we systematically compute the probability linked to the typical user accessing the nearest satellite in tier $i$. 

Simultaneously, we analyze two distinct scenarios and calculate their conditional distance distributions.
\textit{Scenario One: }the $i$-th serving satellite's coordinates fall within ${{\cal A}_{m}^{i}}$.
\textit{Scenario Two: }the $i$-th serving satellite is located in ${{\cal A}_{s}^{i}}$.
Under the conditions of these two scenarios, we derive the corresponding conditional PDF of distances between the typical user and the serving satellite in tier $i$.

\begin{proposition}
The probability of the typical user accessing the nearest satellite in tier $i$ is given by 
\begin{equation}
\label{10}
\begin{aligned}
 &\mathbb{P}_{ass,i}^{E}=\mathbb{P}^{E}\left\{{S = i} \right\} \\
 &=2{\lambda _S^i}\pi \frac{{{R_{S}^{i}} }}{{{R_E}}}\int_{{h_i}}^{r_{max}^i} {r{e^{ - \sum\limits_{k = 1}^K {{\bf{1}}\left( {r - {h_k} \ge 0} \right) {\lambda _S^k}\pi \frac{{{R_{S}^{k}}}}{{{R_E}}}\left( {{r^2} - h_k^2} \right)} }}dr}, 
\end{aligned}
\end{equation}
where $\mathbf{1} \left ( \cdot  \right ) $ is indicator function.
\begin{proof}
    Please refer to Appendix A.
\end{proof}
\end{proposition}
\begin{proposition}
Conditioned on accessing tier $i$, the PDFs of the distance distribution between the typical user and the serving satellite when the serving satellite is in $ {{\cal A}_{m}^{i}}$ and $ {{\cal A}_{s}^{i}}$, respectively, are given by
\begin{equation}
\label{11}
\begin{aligned}
&{f_{{R_i}}^{E}}\left( {r| {\Phi \left( \mathcal{A} _{m}^{i} \right) > 0},S = i} \right) = \\
&\frac{{2{\lambda _S^i}\pi \frac{{R_S^i}}{{{R_E}}}r{e^{ - \sum\limits_{k = 1}^K {{\bf{1}}\left( {r \ge {h_k}} \right){\lambda _S^k}\pi \frac{{R_S^k}}{{{R_E}}}\left( {{r^2} - h_k^2} \right)} }}}}{ \mathbb{P}_{m,ass,i}^{E} },
\end{aligned}
\end{equation}
and
\begin{equation}
\label{12}
\begin{aligned}
&{f_{{R_i}}^{E}}\left( {r| {\Phi \left( \mathcal{A} _{m}^{i} \right) = 0},S = i} \right) = \\
&\frac{{2{\lambda _S^i}\pi \frac{{R_S^i}}{{{R_E}}}r{e^{ - \sum\limits_{k = 1}^K {{\bf{1}}\left( {r \ge {h_k}} \right){\lambda _S^k}\pi \frac{{R_S^k}}{{{R_E}}}\left( {{r^2} - h_k^2} \right)} }}}}{\mathbb{P}_{s,ass,i}^{E}},
\end{aligned}
\end{equation}
where $\mathbb{P}_{m,ass,i}^{E} = \mathbb{P}_{m,i}^{E} \mathbb{P}^{E}_{ass,i}$, $\mathbb{P}_{s,ass,i}^{E} = \mathbb{P}_{s,i}^{E} \mathbb{P}^{E}_{ass,i}$, $\mathbb{P}_{m,i}^{E}$ and $\mathbb{P}_{s,i}^{E}$ are the probabilities that the serving satellite resides in $\mathcal{A} _{m}^{i}$ and $\mathcal{A} _{s}^{i}$, respectively. $\mathbb{P}_{m,i}^{E} = \left( {1 - {e^{ - \lambda _S^i2\pi {{\left( {R_S^i} \right)}^2}\left( {1 - \cos \left( {{\phi _m^i}} \right)} \right)}}} \right)$, and $\mathbb{P}_{s,i}^{E} =  1-\mathbb{P}_{m,i}^{E}$, ${\mathbb{P}^{E}_{ass,i}}$ is given in (\ref{10}).
\end{proposition}
\begin{proof}
    Please refer to Appendix B.
\end{proof}
\subsection{Aggregated Interference Power and Coverage Probability}
In the nearest satellite UAP system,  assuming the serving satellite resides in the tier-$i$, the aggregated interference power follows two scenarios. If the serving satellite resides in $\mathcal{A} _{m}^{i}$, interfering satellites may reside in $\mathcal{A} _{m}^{i}$ and $\mathcal{A} _{s}^{i}$. Conversely, when the serving satellite is within $\mathcal{A} _{s}^{i}$, interfering satellites in the $k$-th tier ($k < i$) reside in $\mathcal{A} _{s}^{i}$, while interfering satellites in the $k$-th tier, (where $k \ge i$) may reside in either $\mathcal{A} _{m}^{i}$ or $\mathcal{A} _{s}^{i}$. Therefore, the
interference power when the serving satellite resides in $\mathcal{A} _{m}^{i}$, is given by
\begin{equation}
\label{13}
\begin{aligned}
&I_{ml}^i = \!\!
 \sum\limits_{k = 1}^K \!\!{\left( \!\!{\sum\limits_{{X'} \in \Phi _m^k\backslash {X_i}} \!\!\!{\!\!\!\!\!{p_k}{H_{{X'}}}{G_{m,k}}\!R_{{X'}}^{ - \alpha }} \! + \!\!\sum\limits_{{X'} \in \Phi _s^k} {\!\!\!{p_k}{H_{{X'}}}{G_{s,k}}R_{{X'}}^{ - \alpha }} } \right)}.
\end{aligned}
\end{equation}
The interference power when the serving satellite resides in $\mathcal{A} _{s}^{i}$, is given by
\begin{equation}
\label{14}
\begin{aligned}
&{I_{sl}^i} \!\!=\!\!\! \sum\limits_{k = i + 1}^K \!{\sum\limits_{{X'} \in \Phi _m^k\backslash {X_i}} {\!\!\!\!\!\!{p_k}{H_{{X'}}}{G_{m,k}}R_{{X'}}^{ - \alpha }} }  \!\!+ \!\!\sum\limits_{k = 1}^K {\sum\limits_{{X'} \in \Phi _s^k} {\!\!\!{p_k}{H_{{X'}}}{G_{s,k}}R_{{X'}}^{ - \alpha }} },
\end{aligned}
\end{equation}
where $X_i$ represents the serving satellite, and ${X^{'}}$ represents interfering satellites in each tier. ${R_{X^{'}}^{-\alpha}}$ represent distances between the typical user and interfering satellites.
\begin{lemma}
Under the condition of accessing the $i$-th tier, the LT of the aggregated interference power when the serving satellite is in $\mathcal{A} _{m}^{i}$ and $\mathcal{A} _{s}^{i}$, respectively, is expressed as
\begin{equation}
\label{15}
\begin{aligned}
&{{\cal L}_{I_{_{ml}}^i}}\left( {s_m^i|\Phi \left( {{\cal A}_m^i} \right) > 0, S = i} \right) = \\
&\exp \!\! \left( { \!\!- 2\pi \! \! \sum\limits_{k = 1}^K {{\! \lambda _S^k}\frac{{R_S^k}}{{{R_E}}}} } \!\!\right.\left(\!\! {\int_{U_{LO}^1\left ( r \right )}^{U_{UP}^1\left ( r \right )} {\!\! \! \!\!\left( {\!1\!\! - \! {{\left( {1 \!+ \! s_m^i{p_k} \! {G_{m,k}}r_I^{ - \alpha } \! \beta } \right)}^{ \! - \chi }}}\! \right)\!{r_I}d{r_I}} } \right. \\
& + 
\left. {\left. {{\rm{          }}\int_{r_m^k}^{r_{\max }^k} {\left( {1 - {{\left( {1 + s_m^i{p_k}{G_{s,k}}r_I^{ - \alpha }\beta } \right)}^{ - \chi }}} \right){r_I}d{r_I}} } \right)} \right),
\end{aligned}
\end{equation}
and 
\begin{equation}
\label{16}
\begin{aligned}
&{{\cal L}_{I_{_{sl}}^i}}\left( {s_s^i| \Phi \left( {{\cal A}_m^i} \right) = 0, S = i} \right) = \\
&\exp\!\! \left( { \!\!- 2\pi \!\!\sum\limits_{k = 1}^K {{\!\lambda _S^k}\frac{{R_S^k}}{{{R_E}}}\!\!\int_{U_{UP}^1 \left ( r \right )}^{U_{UP}^2 \left ( r \right )} \!\!\!\!{\left(\!\! {1 \!  - \!\!{{\left( {1 \!+ \! s_s^i{p_k}{G_{s,k}}r_I^{ - \alpha }\!\beta } \right)}^{ \!- \chi }}} \!\right) \!{r_I}d{r_I}} } }\!\! \right)\\
&\exp \!\!\left(\!\! { - 2\pi \!\!\sum\limits_{k = 1}^K {{\!\lambda _S^k}\!\frac{{R_S^k}}{{{R_E}}}\!\!\int_{U_{LO}^1 \left ( r \right )}^{U_{UP}^1 \left ( r \right )} \!\!{\!\!\!\left( \!{1 \!-\! {{\left( {1 \!+\! s_s^i{p_k}{G_{m,k}}r_I^{ - \alpha }\!\beta } \!\right)}^{ - \chi }}} \!\right)\!{r_I}d{r_I}} } } \!\right),
\end{aligned}
\end{equation}
where $U_{UP}^1\left ( r \right ) ={\rm{max}}\left \{ r_{m}^{k},r \right \} $, $U_{UP}^ 2\left ( r \right ) ={\rm{max}}\left \{ r_{max}^{k},r \right \} $, $U_{LO}^1\left ( r \right ) = {\rm{max}}\left \{ h_{k},r  \right \} $, ${s_m^i} = qA\gamma _{th}^i{\left( {{p_i}{G_{m,i}}} \right)^{ - 1}}{r^\alpha }$, ${s_s^i} = qA\gamma _{th}^i{\left( {{p_i}{G_{s,i}}} \right)^{ - 1}}{r^\alpha }$, $A = \frac{{\chi {{\left( {\chi !} \right)}^{ - \frac{1}{\chi }}}}}{\beta }$.
\end{lemma}
\begin{proof}
  Please refer to Appendix C.  
\end{proof}
\begin{theorem}
Under the paradigm of the nearest satellite UAP, the CP is given in (\ref{17}),
\begin{figure*}[t]
\setlength{\abovecaptionskip}{-1.5cm}
\setlength{\belowcaptionskip}{-0.5cm}
\normalsize
\begin{equation}
\begin{aligned}
&\mathbb{P}_C^E = \sum\limits_{i = 1}^K {\left( \mathbb{P}{\left\{ {{\rm{SINR}}_i^E > \gamma _{th}^i|\Phi \left( {{\cal A}_m^i} \right) > 0,S = i} \right\} \mathbb{P}_{m,ass,i}^{E} } \right.}  + 
\mathbb{P}\left. { \left\{ {{\rm{SINR}}_i^E > \gamma _{th}^i|\Phi \left( {{\cal A}_m^i} \right) = 0,S = i} \right\} \mathbb{P}_{s,ass,i}^{E}} \right)\\
& = \!\!\sum\limits_{i = 1}^K \! {\left( \! {\int_{{h_i}}^{r_m^i} \!\!\! {\Omega \left( {I_{ml}^i,s_m^i} \right) \! 2\lambda _S^i\pi \frac{{R_S^i}}{{{R_E}}}r {e^{\!\!  - \sum\limits_{k \in K} {{\bf{1}}\left( {r \ge {h_k}} \right)\lambda _S^k\pi \frac{{R_S^k}}{{{R_E}}}\left( {{r^2} - h_k^2} \right)} }}dr} } \right.}  \!\! + \! \!
\left. { \int_{r_m^i}^{r_{\max }^i} \!\!\!{\Omega \left( {I_{sl}^i,s_s^i} \right)2\lambda _S^i\pi \frac{{R_S^i}}{{{R_E}}}r{e^{\!\! - \!\sum\limits_{k \in K} {{\! \bf{1}}\left( {r \ge {h_k}} \right)\lambda _S^k\pi \frac{{R_S^k}}{{{R_E}}}\left( {{r^2} - h_k^2} \right)} }}dr} } \right). \label{17}
\end{aligned}
\end{equation}
\hrulefill
\end{figure*} 
where
\begin{equation}
    \begin{aligned}
        \Omega \left( {I,s} \right) = 1 - \sum\limits_{q = 0}^\infty  {\frac{{\chi !}}{{q!\left( {\chi  - q} \right)}}{{\left( { - 1} \right)}^q}{e^{ - s{\sigma ^2}}}} {{\cal L}_I}\left( s \right), \nonumber
    \end{aligned}
\end{equation}
${{\cal L}_{I_{_{ml}}^i}}\!\!\left( {s_m^i|\Phi \!\left( {{\cal A}_m^i} \right) \!>\! 0,S\! = \!i} \right)$, and ${{\cal L}_{I_{_{sl}}^i}} \!\!\left( {s_s^i|\Phi \!\left( {{\cal A}_m^i} \right) \!=\! 0,S\! =\! i} \right)$ are given in Lemma 1.
\end{theorem}
\begin{proof}
  Please refer to Appendix D.  
\end{proof}
\vspace{-2mm}
\subsection{Non-handover Probability}
\captionsetup{font={scriptsize}}
\begin{figure}
\begin{center}
\setlength{\abovecaptionskip}{+0.2cm}
\setlength{\belowcaptionskip}{-0.8cm}
\centering
  \includegraphics[width=2.85in, height=2.0in]{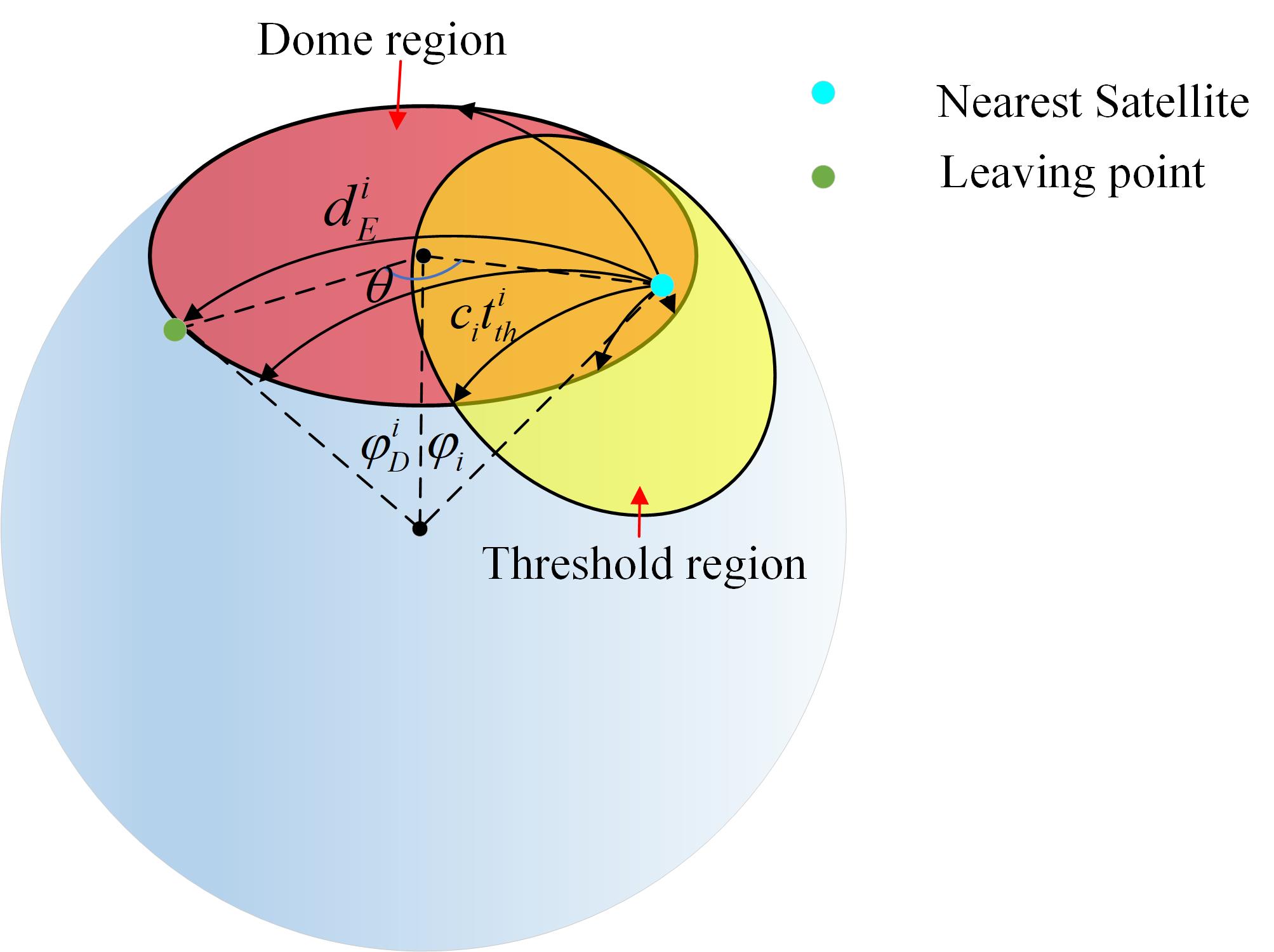}
\renewcommand\figurename{Fig.}
\caption{\scriptsize Graphical representation of the NHP.}
\label{fig: Fig3}
\end{center}
\end{figure}
We define the NHP as the joint probability that a typical user connected to the nearest satellite experiences acceptable communication quality without needing a satellite handover. This probability is defined by the simultaneous event where the typical user receives a signal with an SINR surpassing a threshold, and the time consumed for traversing beyond the dome region of the nearest satellite exceeds the threshold time. 
Let $t_E^i$ denotes the time required for the nearest satellite in the $i$-th tier to exit the dome region ${\varphi _D^i}$ in any direction (see Fig. \ref{fig: Fig3}), representing the switching duration. The threshold time for satellite switching is $t_{th}^i$. Assuming the longitude angle $\theta$ between the satellite's exit direction and its initial motion direction is uniformly distributed with a PDF of $1/ \pi$. By calculating the arc length $d_E^i$ on the spherical surface $R_{S}^{i}$ corresponding to the nearest satellite exiting the dome region and the satellite's velocity $c_{i}$, we determine $t_E^i$.
\begin{theorem}
The NHP of the multi-tier satellite communication system with the nearest satellite UAP is given by
\begin{equation}
\label{18}
\begin{aligned}
&\mathbb{P}_{NH}^E \!=\! \sum\limits_{i = 1}^K {\left( \! {\mathbb{P} \!\left\{ {\mathrm{SINR} _i^E \!\ge \!\gamma _{th}^i,t_E^i \!\ge \! t_{th}^i, \Phi \left( {{\cal A}_m^i} \!\right)\! > \!0| S \!= \!i} \right\}} \right.}\\
&+\left. {\mathbb{P}\left\{ {\mathrm{SINR} _i^E \ge \gamma _{th}^i,t_E^i \ge t_{th}^i, \Phi \left( {{\cal A}_m^i} \right) = 0| S = i} \right\} } \right) \mathbb{P}_{ass,i}^E,
\end{aligned}
\end{equation}
where 
\begin{equation}
\label{19}
\begin{aligned}
&\mathbb{P}\left\{ {{\rm{SINR}}_i^E \ge \gamma _{th}^i,t_E^i \ge t_h^i,\Phi \left( {{\cal A}_m^i} \right) > 0 | S = i} \right\}\\
& =\!  \left\{ \begin{array}{l}
\!\!\! \! \int_{{h_i}}^{r_m^i} {\Omega \left( {I_{ml}^i,s_m^i} \right){f_r^E}\left( {r |S = i} \right)dr}, \ \ \ \ {\rm{if }} \  {p_i}\left( \theta  \right) > 1, \\
\!\!\! \!  \int_{R_{m,dw}^1}^{R_{m,up}^1} {\Omega \left( {I_{ml}^i,s_m^i} \right){f_r^E}\left( {r|S = i} \right)dr}, \ {\rm{if  }} \  \left| {{p_i}\left( \theta  \right)} \right| \le 1, \\
\!\!\! \!  0, \ \ \ \ \ \ \ \ \ \ \ \ \ \ \ \ \ \ \ \ \ \ \ \ \ \ \ \ \ \ \ \ \ \ \ \  \ \ \ \ \ \ {\rm{if }} \ {p_i}\left( \theta  \right) <  - 1,
\end{array} \right.
\end{aligned}
\end{equation}
\begin{equation}
\label{20}
    \begin{aligned}
&\mathbb{P}\left\{ {{\rm{SINR}}_i^E \ge \gamma _{th}^i,t_E^i \ge t_h^i,\Phi \left( {{\cal A}_m^i} \right) = 0 | S = i} \right\}\\
& = \left\{ \begin{array}{l}
\!\!\! \! \int_{{h_i}}^{r_m^i} {\Omega \left( {I_{sl}^i,s_s^i} \right){f_r^E}\left( {r|S = i} \right)dr}, \ \ \ \ {\rm{if }} \ {p_i}\left( \theta  \right) > 1, \\
\!\!\! \! \int_{R_{s,dw}^1}^{R_{s,up}^1} {\Omega \left( {I_{sl}^i,s_s^i} \right){f_r^E}\left( {r| S = i} \right)dr}, \ {\rm{if  }} \ \left| {{p_i}\left( \theta  \right)} \right| \le 1, \\
0, \ \ \ \ \ \ \ \ \ \ \ \ \ \ \ \ \ \ \ \ \ \ \ \ \ \ \ \ \ \ \ \ \ \ \ \  \ \  {\rm{if }} \ {p_i}\left( \theta  \right) <  - 1,
\end{array} \right.
    \end{aligned}
\end{equation}
$\mathbb{P}_{ass,i}^{E}$ is given in Proposition 1, $\left( {R_{m,dw}^1 \le r \le R_{m,up}^1} \right) = \left( {R_{reg}^1 \cup R_{reg}^2} \right) \cap \left( {{h_i} \le r \le r_m^i} \right)$, $\left( {R_{s,dw}^1 \le r \le R_{s,up}^1} \right) = \left( {R_{reg}^1 \cup R_{reg}^2} \right) \cap \left( {r_m^i \le r \le r_{\max }^i} \right)$, $R_{reg}^1 = \left( {{\vartheta ^ + }\left( \theta  \right) < r} \right) \cup \left( {{\varsigma ^ + }\left( \theta  \right) > r} \right) \cup \left( {{\Psi ^ + }\left( \theta  \right) < r} \right)$, $R_{reg}^2 = \left( {{\vartheta ^ - }\left( \theta  \right) < r} \right) \cup \left( {{\varsigma ^ - }\left( \theta  \right) > r} \right) \cup \left( {{\Psi ^ - }\left( \theta  \right) < r} \right)$. ${\vartheta ^ + }\left( \theta  \right)$, ${\varsigma ^ + }\left( \theta  \right) $, ${\Psi ^ + }\left( \theta  \right)$, ${\vartheta ^ - }\left( \theta  \right)$, ${\varsigma ^ - }\left( \theta  \right) $, ${\Psi ^ - }\left( \theta  \right)$ are given in Appendix D. ${f_{{r}}^E} \left ( r\right)$ represents the PDF of the nearest distance distribution in tier $i$. 
\end{theorem}
\begin{proof}
  Please refer to Appendix E.  
\end{proof}
\subsection{Delay Outage Probability}
We define DOP as the probability that the propagation delay for the communication between the serving satellite and the typical user is less than the propagation delay threshold, $T_D^i$, when the user accesses the serving satellite using either the nearest or the max-SINR satellite UAP. In this section, we consider the nearest satellite UAP.
\begin{theorem}
    The DOP of the multi-tier satellite communication system with the nearest satellite UAP is given by
    \begin{equation}
    \label{21}
        \begin{aligned}
&\mathbb{P}_{DO}^E = \sum\limits_{i = 1}^K  \mathbb{P}\left\{ {{R_i} \le cT_D^i|S = i} \right\}\mathbb{P}_{ass,i}^E \\
& = \sum\limits_{i = 1}^K {\left( {1 - \exp \left( { - \lambda _S^i\pi \frac{{R_S^i}}{{{R_E}}}\left( {{{\left( {cT_D^i} \right)}^2} - h_i^2} \right)} \right)} \right)P_{ass,i}^E},
        \end{aligned}
    \end{equation}
    where the $T_{D}^{i}$ is the time propagation delay threshold for the $i$-th tier nearest satellite UAP system, $R_i$ represents the $i$-th nearest distance, $c$ is the speed of light, $\mathbb{P}_{ass,i}^E$ is given in Proposition 1. 
\end{theorem}
\section{Max-SINR Satellite Association Policy in Multi-Tier Satellite Networks}
In this section, we conduct a comprehensive performance analysis of a multi-tier satellite communication network, specifically emphasizing the max-SINR satellite UAP. Our analytical framework is built upon three pivotal components, i.e., an evaluation of association probability, and explorations of CP, NHP, and DOP. These essential elements are seamlessly integrated within this section.
\subsection{Association Probability}
When using the max-SINR satellite UAP, the serving satellite is the satellite that provides the maximum SINR to the typical user. 
In a $K$-tier satellite system, 
the probability that the serving satellite resides in the main-lobe reception region corresponds to the probability that there is at least one satellite in the main-lobe reception region of the dome region associated with $\phi_m^i \left(i = 1,..., K\right)$, which is defined as $\mathbb{P}_m^M = \mathbb{P}\left\{ {\Phi \left( {{{\cal A}_m}} \right) > 0} \right\}$, where ${{\cal A}_m}$ represents the dome region associated with $\phi_m^i \left(i = 1,..., K\right)$. The probability that the serving satellite falls within the side-lobe region is defined as $\mathbb{P}_s^M = \mathbb{P}\left\{ {\Phi \left( {{{\cal A}_m}} \right) = 0} \right\}$.
\begin{proposition}
The probability of the typical user accessing tier $i$ with the max-SINR satellite UAP is given by
\begin{equation}
\label{22}
\begin{aligned}
\mathbb{P}_{ass,i}^{M} = 2{\lambda _S^i}\pi \frac{{R_S^i}}{{{R_E}}}\!\int_{{ h_i}}^{r_{max }^i} \! \! {r{e^{ - \left( {{\psi_i ^{MAX}}\left( r \right) + {\lambda _S^i}\pi \frac{{R_S^i}}{{{R_E}}}\left( {{r^2} - {h_i^2}} \right)} \right)}}dr}, 
\end{aligned}
\end{equation}
where
\begin{equation}
\label{14-1}
\begin{aligned}
&{\psi _i^{MAX}}\left( r \right) = \sum\limits_{ k=1, k \ne i}^{K} {{\lambda _S^k}\pi \frac{{R_S^k}}{{{R_E}}}\varpi \left( {{{\left( {\frac{{{p_k}{G_{m,k}}}}{{{p_i}{G_{m,i}}}}} \right)}^{\frac{1}{\alpha }}}r,{h_k}} \right)}, \\  \nonumber
&\varpi \left( {z,v} \right) = {\bf{1}}\left( {z > v} \right)\left( {{z^2} - {v^2}} \right). \nonumber
\end{aligned}
\end{equation}
Probabilities of the serving satellite residing in the main-lobe reception region and side-lobe reception region, respectively, are given by
\begin{equation}
\label{23}
    \begin{aligned}
    \mathbb{P}_m^M =\mathbb{P}\left\{ {\Phi \left( {{{\cal A}_m}} \right) > 0} \right\} = 1 - \prod\limits_{k = 1}^K {{e^{ - \lambda _S^k2\pi {{\left( {R_S^k} \right)}^2}\left( {1 - \cos \left( {{\phi _m^k}} \right)} \right)}}}, 
    \end{aligned}
\end{equation}
\begin{equation}
\label{24}
    \begin{aligned}
    \mathbb{P}_s^M = \mathbb{P}\left\{ {\Phi \left( {{{\cal A}_m}} \right) = 0} \right\} = \prod\limits_{k = 1}^K {{e^{ - \lambda _S^k2\pi {{\left( {R_S^k} \right)}^2}\left( {1 - \cos \left( {{\phi _m^k}} \right)} \right)}}}. 
    \end{aligned}
\end{equation}
\end{proposition}
\begin{proof}
    Please refer to Appendix F.
\end{proof}

\subsection{Coverage Probability}

\begin{theorem}
With the max-SINR satellite UAP, the comprehensive CP of the system can be articulated as
\begin{equation}
\label{25}
\begin{aligned}
& \mathbb{P}_C^M = \mathbb{P}_m^M \!\! \int_0^\infty \!\!\!\! {\int_{ - \infty }^{ + \infty } \!\! {\frac{{{e^{j\omega y}} \!- \!{e^{j\frac{{\omega y}}{{\max \left( {\delta _{th}^i} \right)}}}}}}{{2\pi j\omega }}{\Theta ^M}\! \left( {\omega ,y} \right) \! {\Upsilon ^M}\! \left( {\omega ,y} \right)d\omega } dy} \\
& + \mathbb{P}_s^M\int_0^\infty  {\int_{ - \infty }^{ + \infty } {\frac{{{e^{j\omega y}} - {e^{j\frac{{\omega y}}{{\max \left( {\delta _{th}^i} \right)}}}}}}{{2\pi j\omega }}{\Theta ^S}\left( {\omega ,y} \right){\Upsilon ^S}\left( {\omega ,y} \right)d\omega } dy},
\end{aligned}
\end{equation}
where $\mathbb{P}_m^M$ and $ \mathbb{P}_s^M$ are given in (\ref{23}), and (\ref{24}), respectively. 
${\Theta ^M}\left( {\omega,y} \right)$, $\Upsilon ^M\left( {\omega,y} \right)$, ${\Theta ^S}\left( {\omega,y} \right)$, and $\Upsilon ^S\left( {\omega,y} \right)$ are given in (\ref{26}),  (\ref{27}), (\ref{28}), and (\ref{29}), respectively \cite{gradshteyn2014table}.   
\begin{figure*}[t]
\setlength{\abovecaptionskip}{-1.5cm}
\setlength{\belowcaptionskip}{-0.5cm}
\normalsize
\begin{equation}
\label{26}
\begin{aligned}
&{\Theta ^M}\left( {\omega ,y} \right) = E\left\{ {{e^{ - j\omega {\sigma ^2}}}} \right\}\exp \left( {\sum\limits_{k = 1}^K { - \lambda _S^k2\pi \frac{{R_S^k}}{{{R_E}}}\int_{r_m^k}^{r_{\max }^k} {\left( {1 - {{\left( {j\omega \beta {p_k}{G_{s,k}}r_I^{ - \alpha } + 1} \right)}^{ - \chi }}} \right){r_I}d{r_I}} } } \right) \times   \\
&\exp \left( {\sum\limits_{k = 1}^K { - \lambda _S^k2\pi \frac{{R_S^k}}{{{R_E}}}\int_{{h_k}}^{r_m^k} {\left( {1 - \frac{{{{\left( {j\omega {p_k}{G_{m,k}}r_I^{ - \alpha } + {\beta ^{ - 1}}} \right)}^{ - \chi }}}}{{{\beta ^\chi }\Gamma \left( \chi  \right)}}\gamma \left( {\chi ,\frac{{j\omega y}}{{\delta _{th}^i}} + \frac{y}{{\beta \delta _{th}^i{p_k}{G_{m,k}}r_I^{ - \alpha }}}} \right)} \right){r_I}d{r_I}} } } \right),
\end{aligned}
\end{equation}
\begin{equation}
\label{27}
    \begin{aligned}
        {\Upsilon ^M}\left( {\omega ,y} \right) = \sum\limits_{k = 1}^K {\frac{{\lambda _S^k2\pi  R_S^k }}{{{R_E}\Gamma \left( \chi  \right)\alpha }}{e^{ - \left( {\frac{{j\omega y}}{{\delta _{th}^i}}} \right)}}{y^{ - \frac{2}{\alpha } - 1}}{{\left( {\beta \delta _{th}^i{p_k}{G_{m,k}}} \right)}^{\frac{2}{\alpha }}}\left( {\gamma \left( {\frac{{\alpha \chi  + 2}}{\alpha },\frac{{y{{\left( {r_{\mathop{\rm m}\nolimits} ^k} \right)}^\alpha }}}{{\delta _{th}^i\beta {p_k}{G_{m,k}}}}} \right) - \gamma \left( {\frac{{\alpha \chi  + 2}}{\alpha },\frac{{y{{\left( {{h_k}} \right)}^\alpha }}}{{\delta _{th}^i\beta {p_k}{G_{m,k}}}}} \right)} \right)}, 
    \end{aligned}
\end{equation}
\begin{equation}
\label{28}
    \begin{aligned} 
   & {\Theta ^S}\left( {\omega ,y} \right) = E\left\{ {{e^{ - j\omega {\sigma ^2}}}} \right\}\exp \left( {\sum\limits_{k = 1}^K { - \lambda _S^k2\pi \frac{{R_S^k}}{{{R_E}}}\int_{{h_k}}^{r_m^k} {\left( {1 - {{\left( {j\omega \beta {p_k}{G_{m,k}}r_I^{ - \alpha } + 1} \right)}^{ - \chi }}} \right){r_I}d{r_I}} } } \right) \times \\
   & \exp \left( {\sum\limits_{k = 1}^K { - \lambda _S^k2\pi \frac{{R_S^k}}{{{R_E}}}\int_{r_m^k}^{r_{\max }^k} {\left( {1 - \frac{{{{\left( {j\omega {p_k}{G_{s,k}}r_I^{ - \alpha } + {\beta ^{ - 1}}} \right)}^{ - \chi }}}}{{{\beta ^\chi }\Gamma \left( \chi  \right)}}\gamma \left( {\chi ,\frac{{j\omega y}}{{\delta _{th}^i}} + \frac{y}{{\beta \delta _{th}^i{p_k}{G_{s,k}}r_I^{ - \alpha }}}} \right)} \right){r_I}d{r_I}} } } \right),
    \end{aligned}
\end{equation}
\begin{equation}
\label{29}
    \begin{aligned}
        {\Upsilon ^S}\left( {\omega ,y} \right) = \sum\limits_{k = 1}^K { \frac{{\lambda _S^k2\pi R_S^k}}{{{R_E}\Gamma \left( \chi  \right) \alpha }}{e^{ - \left( {\frac{{j\omega y}}{{\delta _{th}^i}}} \right)}}{y^{ - 1 - \frac{2}{\alpha }}}{{\left( {\beta \delta _{th}^i{p_k}{G_{s,k}}} \right)}^{\frac{2}{\alpha }}}\left( {\gamma \left( {\frac{{\alpha \chi  + 2}}{\alpha },\frac{{u{{\left( {r_{\max }^k} \right)}^\alpha }}}{{\delta _{th}^i\beta {p_k}{G_{s,k}}}}} \right) - \gamma \left( {\frac{{\alpha \chi  + 2}}{\alpha },\frac{{u{{\left( {r_m^k} \right)}^\alpha }}}{{\delta _{th}^i\beta {p_k}{G_{s,k}}}}} \right)} \right)}. 
    \end{aligned}
\end{equation}
\hrulefill
\end{figure*}
\end{theorem}
\begin{proof}
  Please refer to Appendix G.  
\end{proof}

\subsection{Non-handover Probability}
\begin{theorem}
The NHP of the multi-tier satellite communication system with max-SINR satellite UAP is given by
\begin{equation}
\label{30}
\begin{aligned}
&\mathbb{P}_{NH}^M = \sum\limits_{i = 1}^K {\left( {\mathbb{P}\left\{ {{\rm{SINR}}_i^E \ge \gamma _{th}^i,t_E^i \ge t_{th}^i,\Phi \left( {{\cal A}_m^i} \right) > 0|S = i} \right\}} \right.} \\
& + \left. {\mathbb{P}\left\{ {{\rm{SINR}}_i^E \ge \gamma _{th}^i,t_E^i \ge t_{th}^i,\Phi \left( {{\cal A}_m^i} \right) = 0|S = i} \right\}} \right) \mathbb{P}_{ass,i}^M,
\end{aligned}
\end{equation}
where $\mathbb{P}_{ass,i}^{M}$ is given in Proposition 3, ${\mathbb{P}\left\{ {\mathrm{SINR} _i^E \ge \gamma _{th}^i,t_E^i \ge t_{th}^i,\Phi \left( {{\cal A}_m^i} \right) > 0|S = i} \right\}}$, and ${\mathbb{P}\left\{ {\mathrm{SINR} _i^E \ge \gamma _{th}^i,t_E^i \ge t_{th}^i,\Phi \left( {{\cal A}_m^i} \right) = 0|S = i} \right\}}$ are given in Theorem 2.
\end{theorem}

\begin{theorem}
    The DOP of the multi-tier satellite communication system with max-SINR satellite UAP is given by
    \begin{equation}
    \label{31}
        \begin{aligned}
&\mathbb{P}_{DO}^M = \sum\limits_{i = 1}^K  \mathbb{P}\left\{ {{R_i} \le cT_D^i|S = i} \right\}\mathbb{P}_{ass,i}^M\\
& = \sum\limits_{i = 1}^K {\left( {1 - \exp \left( { - \lambda _S^i\pi \frac{{R_S^i}}{{{R_E}}}\left( {{{\left( {cT_D^i} \right)}^2} - h_i^2} \right)} \right)} \right)P_{ass,i}^M},
        \end{aligned}
    \end{equation}
   $\mathbb{P}_{ass,i}^M$ is given in Proposition 3. 
\end{theorem}
\section{Simulations and Discussions}
In this section, Monte Carlo simulations are conducted to evaluate the CP and NHP of a typical user in the $K$-tier downlink LEO satellite networks.
Default parameters are listed in Table \ref{Table2}, unless otherwise specified. Fig. 4 illustrates a 3D diagram for multi-tier HetSatNet with the PPP constellation and Starlink constellation, respectively. In order to ensure the tractability of modeling and the feasibility of simulation, a simplified version of the Starlink constellation was used. The LEO satellites are distributed at different orbital altitudes, based on a homogeneous SPPP model.
\begin{table}[t]
\setlength{\abovecaptionskip}{0cm}
\captionsetup{font={normalsize}}
\caption{Default Parameters Setup}
\label{Table2}
\begin{center}
\small
 \begin{tabular}{c c}
 \hline
 \hline
 \textbf{Parameter} &  \textbf{Value} \\ 
 \hline
 \hline
 $h_1$,$h_2$,$h_3$ &  600 km, 900 km, 1200 km \cite{9678973}\\
 
 $N_S^1$,$N_S^2$,$N_S^3$ &   $100$ \cite{wang2023coverage}\\

 $P_1$,$P_2$,$P_3$ & 10 dBW \cite{okati2020stochastic}, 15 dBW \cite{wang2023coverage}, 30 dBW \cite{song2022cooperative}\\
 
$G_{ml,1}$,$G_{ml,2}$,$G_{ml,3}$ &  47 dBi \cite{guo2024user}\\

 $G_{sl,1}$,$G_{sl,2}$,$G_{sl,3}$ & 20 dBi \cite{guo2024user}\\
 
 $c_1$,$c_2$,$c_3$ & $7.5997e^3$ m/s,$7.5887e^3$ m/s,$7.5723e^3$ m/s \\
 
 $\sigma^2 $ & $1e^{-12}$\\
 
$\varphi_{D}^1$,$\varphi_{D}^2$,$\varphi_{D}^3$  & $0.1$ rads, $0.1$ rads, $0.1$ rads \\
 
 $\alpha$, $m$, $b_0$, $\Omega$ & $2$, $2$, $1$, $1$ \\
  $T_D^1$, $T_D^2$, $T_D^3$ & 10 ms, 10 ms, 10 ms\\

  $\varphi_m^i$ & 0.01 rads \\
 \hline
 \hline
\end{tabular}
\vspace{-0.4em} 
\end{center}
\end{table}

\subsection{Coverage Probability Simulations}

\captionsetup{font={scriptsize}}
\begin{figure}[t]
\begin{center}
\setlength{\abovecaptionskip}{-0.0cm}
\setlength{\belowcaptionskip}{-0.9cm}
\centering
  \includegraphics[width=3.4in, height=1.9in]{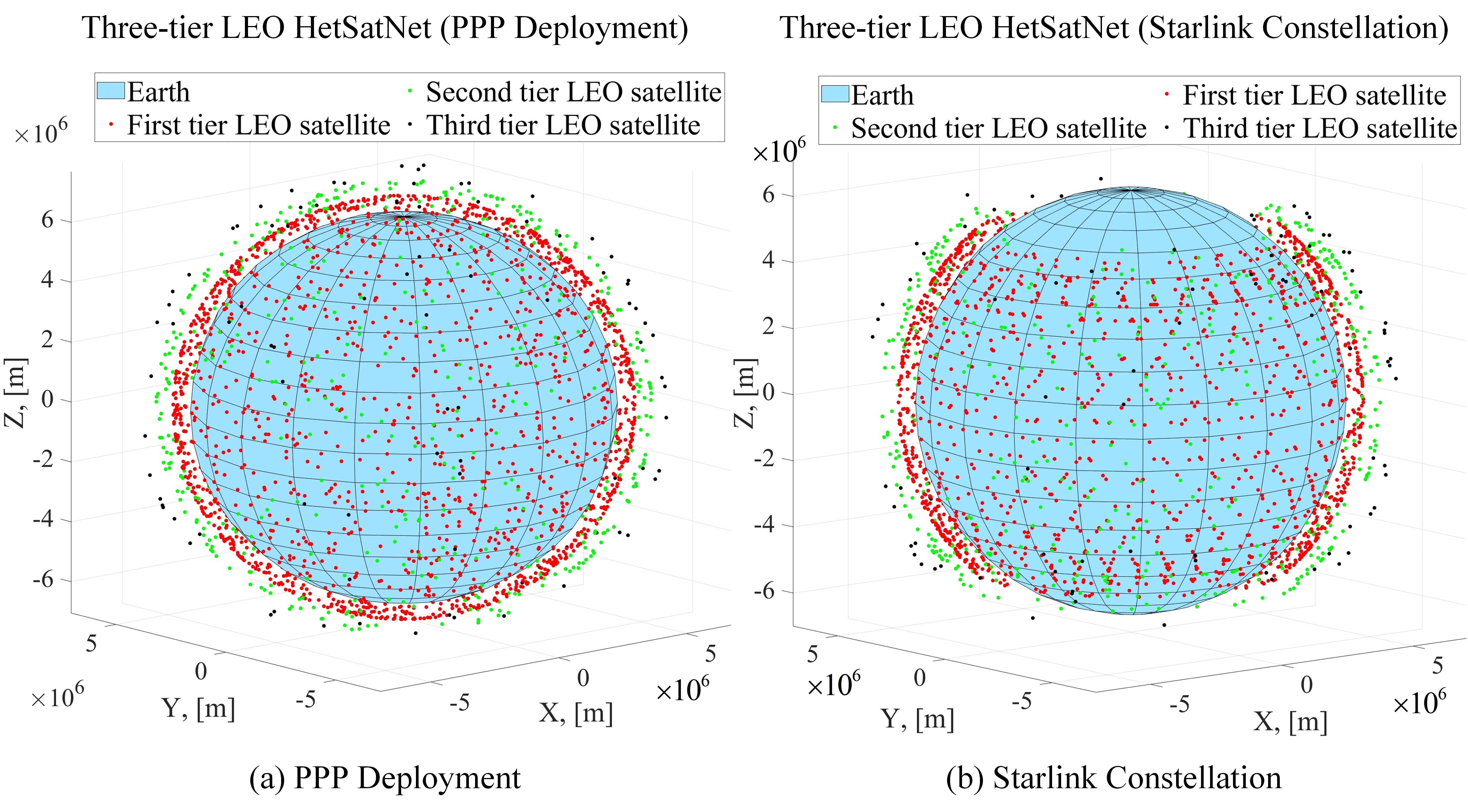}
\renewcommand\figurename{Fig.}
\caption{\scriptsize Three-dimensional graph of a three-tier LEO HetSatNet with PPP constellation and Starlink constellation.}
\label{Fig 4}
\end{center}
\end{figure}
Fig. \ref{Fig 5} delineates the CP of 1-tier, 2-tier, and 3-tier satellite systems under two UAPs, i.e., nearest and max-SINR satellite UAPs. 
Observations from Fig. \ref{Fig 5} indicate that, within the max-SINR UAP framework, the CPs of the multi-tier satellite systems significantly exceed those of the nearest UAP system. Remarkably, when the SINR threshold is configured at $-5$ dB, the highest CP attained by a 1-tier system using the nearest UAP is 0.481, 0.77, and 0.91 for the 1-tier, 2-tier, and 3-tier systems, respectively. These CP values demonstrate substantial performance gains compared to the nearest UAP system. To summarize, the communication system governed by the max-SINR UAP exhibits superior performance relative to its nearest UAP counterpart when the higher-tier satellites have relatively more extensive powers. Furthermore, within the max-SINR UAP framework, the hierarchy of CPs indicates that the 3-tier system surpasses the 2-tier system, which outperforms the 1-tier system under certain parameter configurations. To obtain more design insights in actual situations, we will explore the optimal density of satellites in each tier when the system CP is maximized in a two-tier system. This will allow us to deploy the optimal number of satellites to achieve optimal coverage in the real-world scenarios.
\captionsetup{font={scriptsize}}
\begin{figure}[t]
\begin{center}
\setlength{\abovecaptionskip}{-0.0cm}
\setlength{\belowcaptionskip}{-0.5cm}
\centering
  \includegraphics[width=3.4in, height=2.8in]{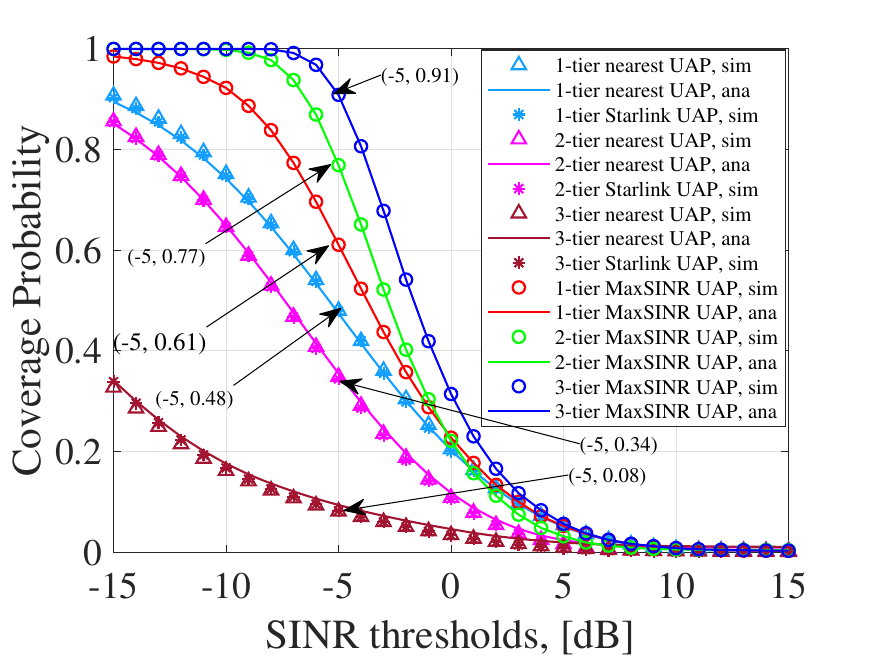}
\renewcommand\figurename{Fig.}
\caption{\scriptsize CP varies across various SINR thresholds, with two UAPs for 1-tier network, 2-tier network and 3-tier network. (sim: simulation, ana: analysis)}
\label{Fig 5}
\end{center}
\end{figure}
\captionsetup{font={scriptsize}}
\begin{figure}[t]
\begin{center}
\setlength{\abovecaptionskip}{-0.0cm}
\setlength{\belowcaptionskip}{-0.5cm}
\centering
  \includegraphics[width=3.4in, height=2.8in]{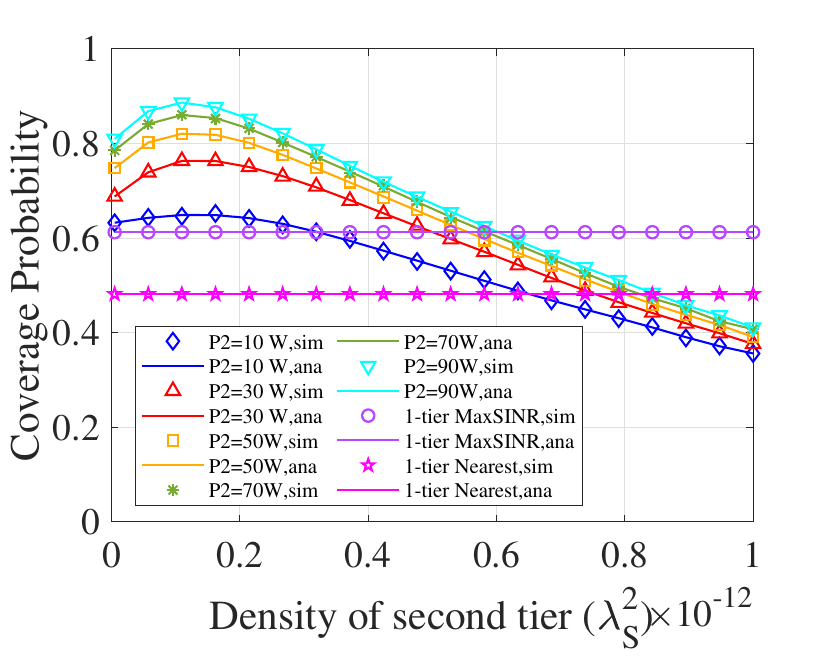}
\renewcommand\figurename{Fig.}
\caption{\scriptsize CPs with max-SINR UAP and nearest UAP across a spectrum of satellite density of second tier in a 2-tier system, evaluated under a variety of the second-tier satellite transmit power (-5 dB).}
\label{Fig 6}
\end{center}
\end{figure}

\captionsetup{font={scriptsize}}
\begin{figure}[t]
\begin{center}
\setlength{\abovecaptionskip}{-0.0cm}
\setlength{\belowcaptionskip}{-0.5cm}
\centering
  \includegraphics[width=3.2in, height=2.8in]{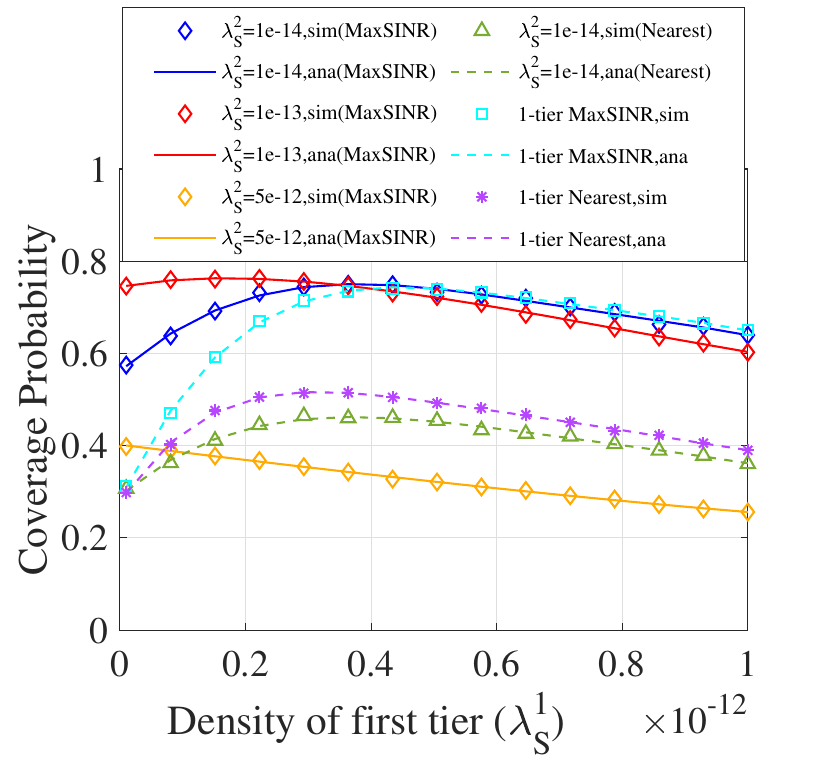}
\renewcommand\figurename{Fig.}
\caption{\scriptsize CPs with max-SINR UAP and nearest UAP across a spectrum of satellite density of first tier in 2-tier system, evaluated under a variety of second-tier satellite satellite density (-5 dB).}
\label{Fig 7}
\end{center}
\end{figure}

Fig. \ref{Fig 6} illustrates the CP dynamics in a 2-tier max-SINR UAP scheme with a SINR threshold of $-5$ dB. The graphic reveals the intricate relationship between the density of satellites in the second tier and satellite transmission power. Specifically, the graphical representation delineates the CP dynamics of the 2-tier max-SINR UAP system, considering distinct levels of second-tier satellite transmission power, i.e., $10 W$, $30 W$, $50 W$, $70 W$, and $90 W$. Noteworthy patterns emerge, revealing optimal satellite density configurations for each transmission power setting, leading to the attainment of maximal CP. For instance, when the second-tier satellite transmit power is set to $30 W$, the system achieves a maximum CP of 0.76 at an optimal second-tier satellite density of $1.62\times 10^{-13}$, with the corresponding first-tier satellite density being $1.64 \times 10^{-13}$.
This optimal performance can be attributed to the relatively sparse deployment of first-tier satellites, users predominantly connect to second-tier satellites. As the second-tier satellite density increases, the likelihood of connecting to a closer second-tier satellite improves, leading to stronger received signal power and thus an increase in CP. However, further increases in second-tier satellite density introduce excessive co-channel interference, which outweighs the benefits of signal enhancement and results in a gradual decline in CP.

Concurrently, a comparative analysis is conducted under consistent parameter settings to distinguish the performance differentials among the 2-tier max-SINR UAP system, the 1-tier max-SINR UAP system, and the 1-tier nearest UAP system. The graphical depiction accentuates that the performance of the 2-tier max-SINR UAP system with $P_{2}= 10 W$  is comparatively diminished in relation to the 1-tier max-SINR UAP system with $P_{1}= 10 W$ when the second-tier density exceeds $0.4 \times 10^{-12}$. 

Fig.~\ref{Fig 7} presents a detailed evaluation of the CP in a two-tier LEO satellite network. The analysis investigates the joint influence of first-tier and second-tier satellite densities on CP, assuming fixed transmit powers of 10 dBW and 15 dBW for the first and second tiers, respectively. The density of the first-tier satellites varies from $1\times 10^{-14} /m^2$ to $1\times 10^{-12} /m^2$.
Under the 1-tier Max-SINR UAP configuration, the CP initially increases with the first-tier satellite density, then decreases. This non-monotonic behavior arises from a trade-off between improved signal strength and accumulated interference. When the density increases, the typical user is more likely to connect to a closer satellite with stronger signal power. However, excessive density introduces significant co-channel interference, eventually degrading the SINR and thus the CP.

For the two-tier system, when the second-tier satellite density is moderate, a similar non-monotonic trend is observed. Initially, due to the sparsity of first-tier satellites, users tend to associate with second-tier satellites, which, despite higher path loss, offer stronger signals owing to their higher transmit power. As the first-tier density increases, the received signal power from first-tier satellites improves, making them more competitive in SINR. Consequently, more users begin associating with the first tier, while the availability of second-tier satellites provides fallback opportunities in case of weak or obstructed first-tier links. This dual-access capability enhances the overall CP. However, beyond a certain density threshold, interference from both tiers dominates, and the CP begins to deteriorate, as shown by red and blue line in Fig. \ref{Fig 7}. When the second-tier satellite density is high, the network becomes interference-limited even at the early stages. The increased number of interfering second-tier satellites lowers the CP, and further increases in the first-tier density exacerbate interference, leading to a continuous decline in CP, as reflected by the orange curve in Fig. \ref{Fig 7}.

Comparing the two-tier and one-tier systems in Fig. \ref{Fig 7}, we observe that in the low-density regime, the two-tier system outperforms the single-tier system due to a broader selection of potential serving satellites. As the first-tier density increases, users benefit from shorter communication distances, and the Max-SINR policy enables optimal tier selection. Additionally, the presence of second-tier satellites provides an auxiliary layer of connectivity, thereby further improving CP. Eventually, as the first-tier density becomes excessive, the single-tier system, which is subject to interference from only one layer, slightly outperforms the two-tier system, which suffers from cumulative interference from both tiers. This is evident from the intersection of the red and light-blue curves.

Under the Nearest satellite UAP, a similar trend is noted, though with key distinctions. Since users connect to the nearest satellite regardless of signal quality, initial CP performance is limited by the sparsity of the first tier. As the density increases, CP improves due to enhanced signal strength. However, the absence of signal-based association causes the second-tier satellites to act primarily as interferers. As a result, the two-tier Nearest UAP system consistently underperforms its single-tier counterpart.

In summary, the Max-SINR UAP significantly improves CP relative to the Nearest UAP. Moreover, under appropriate parameter settings, the proposed multi-tier satellite architecture yields notable performance gains over conventional single-tier deployments. These results offer valuable design guidelines for practical multi-tier HetSatNet systems, particularly in identifying optimal first-tier and second-tier satellite density configurations to maximize CP.

From a complexity perspective, the Nearest UAP is simpler to implement, as users associate with the closest satellite in physical distance. In contrast, the Max-SINR UAP requires users to compute and compare SINR values across all visible satellites before selecting the optimal link, resulting in higher computational overhead. Therefore, a trade-off exists between complexity and performance:, i.e., Max-SINR offers higher CP at the cost of increased processing, while Nearest UAP provides implementation simplicity with slightly reduced performance. Striking an appropriate balance between these factors is critical for system-level optimization.

\subsection{Non-handover Probability Simulations}
\captionsetup{font={scriptsize}}
\begin{figure}[t]
\begin{center}
\setlength{\abovecaptionskip}{-0.0cm}
\setlength{\belowcaptionskip}{-0.9cm}
\centering
  \includegraphics[width=3.2in, height=2.8in]{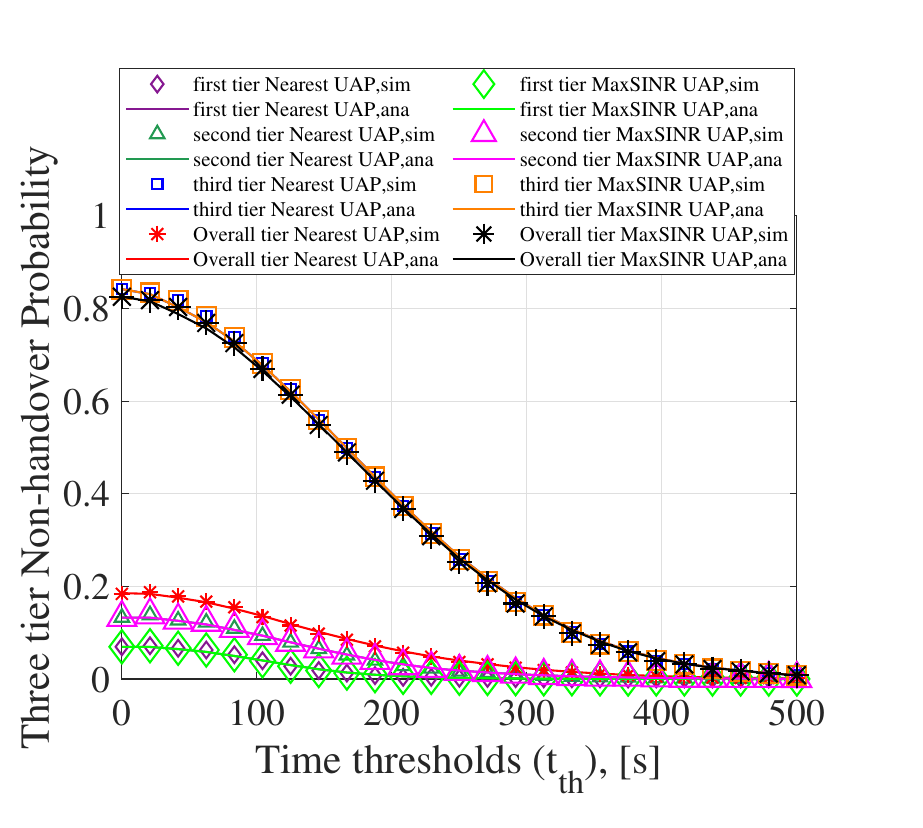}
\renewcommand\figurename{Fig.}
\caption{\scriptsize NHP across various time thresholds, evaluated each tier for a 3-tier network under two UAPs (Nearest and Max-SINR) (-10 dB). }
\label{Fig 8}
\end{center}
\end{figure}
Fig. \ref{Fig 8} conducts a detailed analysis of how the NHP changes over time by applying the nearest and max-SINR UAPs, specifically in the context of a SINR threshold set at -10 dB. It is evident that when accessing each tier, the NHP curves for both the nearest UAP and max-SINR UAP systems overlap as the time thresholds vary.
Exemplified by the alignment of the NHP curve for the first tier under the nearest UAP system with that of the first tier under the max-SINR UAP system, this convergence stems from the inherent property that, assuming access to a specific tier, NHP remains contingent solely upon the temporal requisites for the nearest satellite in that tier to traverse the dome region. This condition renders NHP invariant with respect to the chosen UAP under the assumption of accessing a specific tier. However, owing to distinct association probabilities for each tier, the overall NHP manifests different paths between the two policies.

Furthermore, Fig. \ref{Fig 8} illuminates that within the nearest UAP system, the overall NHP closely aligns with the NHP value of the first tier, thereby affirming the predominant influence of the first-tier satellite's NHP on the overall NHP of the system. In contrast, within the max-SINR UAP system, the NHP of the third-tier satellite remains in a paramount position. Additionally, Fig. \ref{Fig 8} shows that the max-SINR UAP system can attain a larger NHP than the nearest UAP. Particularly under prolonged time thresholds, the max-SINR UAP system demonstrates the capability to curtail switching probabilities while concurrently preserving favorable communication quality, thereby distinguishing itself from the nearest UAP system.

\captionsetup{font={scriptsize}}
\begin{figure}[t]
\begin{center}
\setlength{\abovecaptionskip}{-0.0cm}
\setlength{\belowcaptionskip}{-0.0cm}
\centering
 \includegraphics[width=3.2in, height=2.8in]{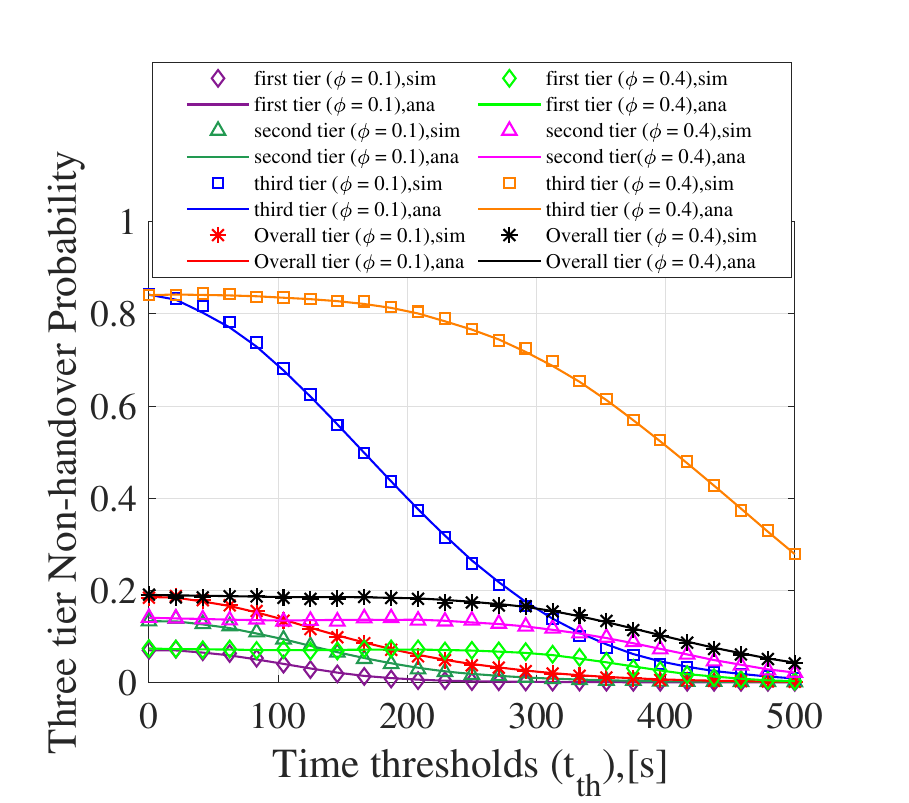}
\renewcommand\figurename{Fig.}
\caption{\scriptsize 3 tiers NHP: Variations across time thresholds and the Earth-based contact angles for dome region with nearest UAP when the SINR threshold is -10 dB.}
\label{Fig 10}
\end{center}
\end{figure}

Fig. \ref{Fig 10} delineates the variation in NHP for each tier and the aggregate system in response to temporal adjustments, considering distinct Earth-based contact angles for the dome region. As the contact angle increases, a noticeable shift occurs in the temporal domain, where NHP undergoes a substantial decrease. This phenomenon can be attributed to the increased temporal requirement for the nearest satellite to exit the dome region. Consequently, the probability that the time required for the nearest satellite to leave the dome region exceeds the time threshold rises, leading to an increase in NHP values.
\captionsetup{font={scriptsize}}
\begin{figure}[t]
\begin{center}
\setlength{\abovecaptionskip}{-0.0cm}
\setlength{\belowcaptionskip}{-0.5cm}
\centering
  \includegraphics[width=3.2in, height=2.8in]{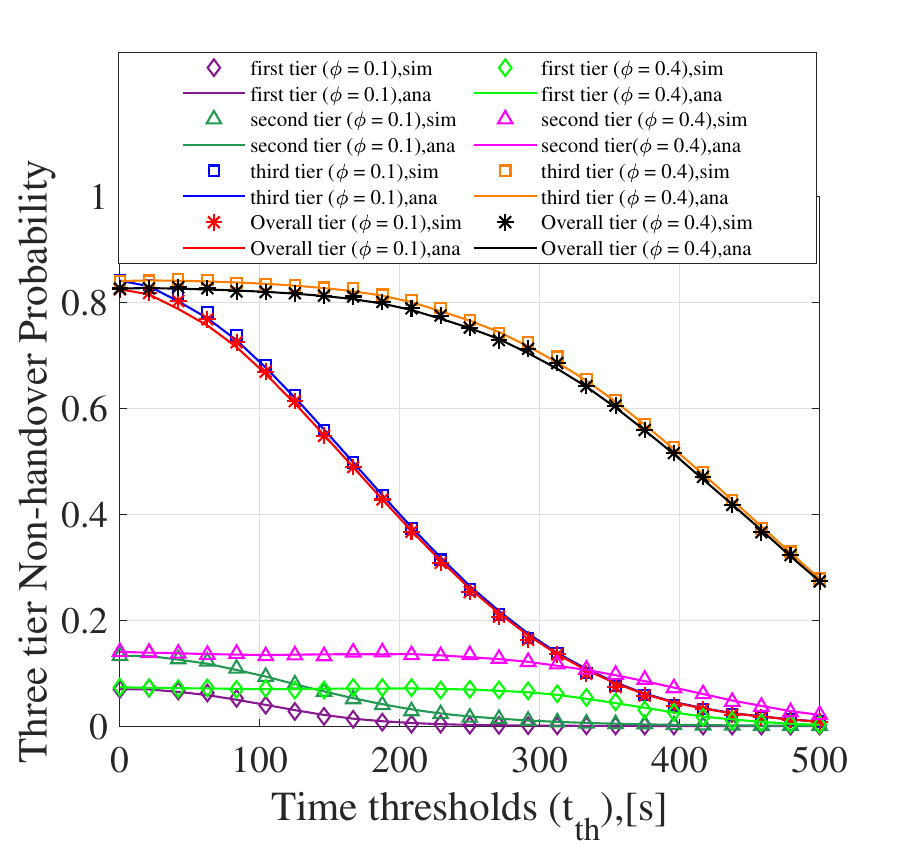}
\renewcommand\figurename{Fig.}
\caption{\scriptsize 3 tiers NHP: Variations across time thresholds and the Earth-based contact angles for dome region with max-SINR UAP when the SINR threshold is -10 dB.}
\label{Fig 11}
\end{center}
\end{figure}

Fig. \ref{Fig 11} illustrates the variation of NHP for each tier and the entire system in a max-SINR UAP system based on the various contact angles of the dome region. Comparing Fig. \ref{Fig 10} and Fig. \ref{Fig 11},  it can be observed that the overall NHP in the max-SINR system changes more gradually with an increasing time threshold. This gradual decline trend in real-world satellite design can help mitigate the issue of signal quality degradation caused by rapid satellite handovers. Consequently, the system can maintain a more stable connection, thereby reducing the frequency and severity of signal interruptions. This approach can enhance the overall reliability and performance of satellite communications networks.
\captionsetup{font={scriptsize}}
\begin{figure}[t]
\begin{center}
\setlength{\abovecaptionskip}{-0.2cm}
\setlength{\belowcaptionskip}{-0.5cm}
\centering
  \includegraphics[width=3.2in, height=2.8in]{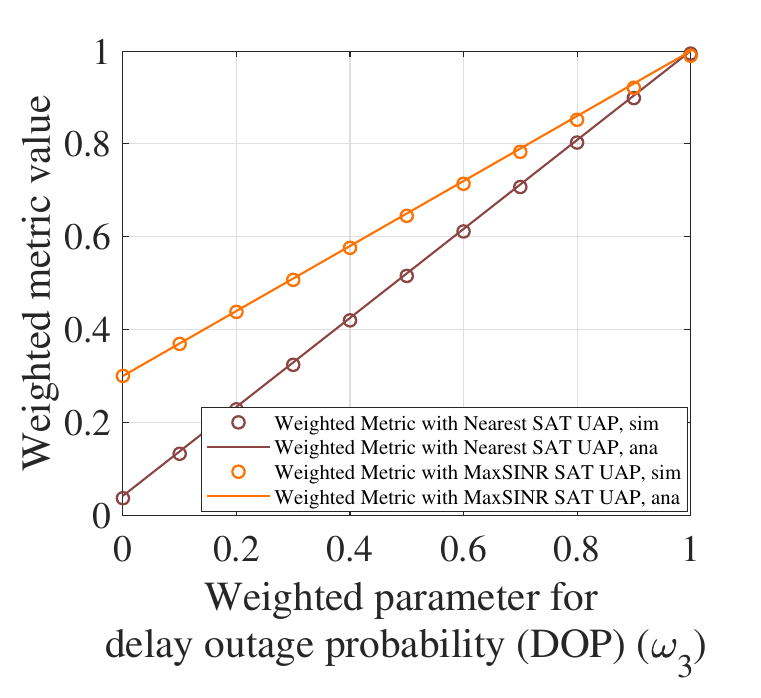}
\renewcommand\figurename{Fig.}
\caption{\scriptsize Weight Metric: Variations across the weighted parameter of time delay outage $\omega_3^{E,M}$ (0 dB).}
\label{Fig 12}
\end{center}
\end{figure}

To further illustrate the practical applicability of the WM framework, we conducted a systematic parameter sweep over the weight space using a meshgrid-based scanning approach. Specifically, $\omega_1^{N,M}$ and $\omega_3^{N,M}$ were varied within [0,1], while ensuring $\omega_2^{N,M} = 1- \omega_1^{N,M} - \omega_3^{N,M}$. A validity mask was applied to satisfy the non-negative constraint for all weights. The results shown in Fig. \ref{Fig 12} provide a performance landscape that reveals the robustness of the Max-SINR UAP across different weight configurations. Particularly, we observed that the Max-SINR UAP consistently outperforms the Nearest UAP across a wide range of weight settings. The primary reason for this observation lies in the specific delay time threshold settings used in our simulations. Under these configurations, both the Max-SINR UAP and the Nearest UAP tend to exhibit relatively large DOP values, regardless of the UAP. As a result, the influence of DOP on the WM becomes less significant in distinguishing between the two strategies. Meanwhile, Max-SINR UAP benefits from better CP and NHP performance due to its stronger signal quality and interference management, which consistently leads to higher WM values across the scanned weight space. This highlights the critical role of WM scanning as not just a theoretical model but a practical design tool to assist network operators in selecting optimal UAP strategies under varying service profile.

By considering a weighted metric that includes CP, NHP, and DOP, and analyzing the impact of each factor on overall system performance, we observed a trade-off among these parameters. This insight provides valuable guidance for the practical design of multi-layer HetSatNet, i.e., depending on the specific scenario, different weights should be selected and different UAPs should be employed. For communication scenarios requiring minimal transmission delay, $\omega_3^{E,M}$ should be larger. Conversely, for scenarios demanding low satellite handover frequency and higher communication quality, the max-SINR satellite UAP is more suitable, $\omega_1^{E,M}$, and $\omega_2^{E,M}$ should be larger.
\vspace{-1em}
\section{Conclusions}
This work presented a multi-tier satellite-terrestrial communication system operating under two UAPs, i.e., nearest satellite and max-SINR satellite UAP. We computed the CP, NHP, DOP, and a novel weighted metric formed by the weighted amalgamation of CP, NHP, and DOP for both UAPs.
In addition, we compared the CP, NHP, DOP, and weighted metric between two UAP systems, revealing advantages for the nearest UAP and the max-SINR UAP in different scenarios. 
These findings provide valuable insights into real-world design: the max-SINR UAP system can improve signal reception quality and reduce the frequency of satellite handovers compared with the nearest satellite UAP. However, the nearest satellite UAP can provide a shorter time delay. Our findings confirm that comparing WM under different weights provides not only theoretical insight but also practical guidance for strategy selection in real-world deployments. The flexibility of the WM framework supports a systematic approach for balancing coverage, stability, and delay requirements across a variety of satellite network scenarios.

\begin{appendices}
\vspace{-3mm}

\section{Proof of Proposition 1}

    The probability of the typical user accessing
the nearest satellite in tier i is derived based on \cite{andrews2016primer}.
\begin{equation}
    \begin{aligned}
& \mathbb{P}_{ass,i}^E = {\mathbb{P}^E}\left\{ {S = i} \right\}\\
 &= \int_{{h_i}}^{r_{\max }^i} {\mathbb{P}\left\{ {{R_k} > r,k \ne i} \right\}{f_{{R_i}}}\left( r \right)dr} \\
& = \int_{{h_i}}^{r_{\max }^i} {\prod\limits_{k \ne i} {{e^{ - \lambda _S^k\pi \frac{{R_S^k}}{{{R_E}}}\left( {{r^2} - h_k^2} \right)}}} {f_{{R_i}}}\left( r \right)dr} \\
& = \int_{{h_i}}^{r_{\max }^i} {{e^{ - \sum\limits_{k \ne i} {{\bf{1}}\left( {{r^2} - h_k^2 \ge 0} \right)\lambda _S^k\pi \frac{{R_S^k}}{{{R_E}}}\left( {{r^2} - h_k^2} \right)} }}{f_{{R_i}}}\left( r \right)dr} \\
 &= 2\lambda _S^i\pi \frac{{R_S^i}}{{{R_E}}}\int_{{h_i}}^{r_{\max }^i} {r{e^{ - \sum\limits_{k \in K} {{\bf{1}}\left( {r - {h_k} \ge 0} \right)\lambda _S^k\pi \frac{{R_S^k}}{{{R_E}}}\left( {{r^2} - h_k^2} \right)} }}dr} 
    \end{aligned}
    \end{equation}

\section{Proof of Proposition 2}
Within the framework of the nearest satellite UAP system, the conditional PDF of distance when the serving satellite resides in $i$-th tier main-lobe reception region can be derived from the cumulative
distribution function (CDF).
\label{32}
\begin{align}
&{f_{{R_i}}^E}\left( {r|\Phi \left( {{{\cal A}_m^i}} \right) > 0,S = i} \right) \nonumber \\
& = \frac{d}{{dr}}\left( {1 - \frac{{\mathbb{P}\left\{ {{R_i} > r,\Phi \left( {{\cal A}_m^i} \right) > 0,S = i} \right\}}}{{\mathbb{P}\left\{ {\Phi \left( {{\cal A}_m^i} \right) > 0,S = i} \right\}}}} \right) \nonumber\\
&\mathop  = \limits^{\left( a \right)} \frac{d}{{dr}}\left( {1 - \frac{{\int_r^{r_m^i} {\prod\limits_{k \ne i} {{e^{ - \lambda _S^k\pi \frac{{R_S^k}}{{{R_E}}}\left( {{u^2} - h_k^2} \right)}}} } {f_{{R_i}}}\left( u \right)du}}{{\mathbb{P}_{m,ass,i}^E}}} \right) \nonumber \\
& = \frac{{2\lambda _S^i\pi \frac{{R_S^i}}{{{R_E}}}r{e^{ - \sum\limits_{k = 1}^K {\left( {r \ge {h_k}} \right)\lambda _S^k\pi \frac{{R_S^k}}{{{R_E}}}\left( {{r^2} - h_k^2} \right)} }}}}{{\mathbb{P}_{m,ass,i}^E}}.
\end{align}
where ${f_{{R_i}}} \left ( u\right)$ represents the PDF of the nearest distance distribution in tier $i$. Similarly, we can get ${f_{{R_i}}^{E}}\left( {r| {\Phi \left( \mathcal{A} _{m}^{i} \right) = 0},S = i} \right)$, which is given in (\ref{12}).
\section{Proof of Lemma 1}
The LTs of the interference power when the serving satellite resides in the main-lobe and side-lobe reception regions of the $i$-th tier are given by
\label{33}
\begin{align}
&{{\cal L}_{I_{_{ml}}^i}}\left( {s_m^i|\Phi \left( {{\cal A}_m^i} \right) > 0,S = i} \right) \nonumber = \mathbb{E}\left\{ {{e^{ - s_m^i I_{_{ml}}^i}}} \right\}  \nonumber \\
&\mathop  = \limits^{\left( a \right)} \!\! \exp \!\!\left ( \!\! -2\pi \!\!\sum_{k=1}^{K} \!\!\left(\!\! {\lambda _S^k\frac{{R_S^k}}{{{R_E}}}\!\!\int_{U_{LO}^1\left ( r \right )}^{U_{UP}^1\left ( r \right )}\!\!\!\!\! {\left(\! {1\!\!-\!\! {{\left( {1 \!\! + \! s_m^i {p_k}{G_{m,k}}r_I^{ - \alpha }\!\beta \!} \right)}^{\! - \chi }}} \!\right)\!\!{r_I}\!d{r_I}} } \!\!\right) \!\! \right )  \nonumber\\
& \times \!\!\exp \!\! \left ( \!\! -2\pi \!\! \sum_{k=1}^{K}  \lambda _S^k\frac{{R_S^k}}{{{R_E}}} \!\! \int_{r_m^k}^{r_{\max }^k} \!\! \!\! {\left(\!\! {1 \!-\! {{\left( {1 \! + \! s_m^i {p_k}{G_{s,k}}r_I^{ - \alpha } \! \beta } \right)\! }^{\! - \chi }}} \! \right)\! {r_I}d{r_I}}  \right ), \nonumber \\
\end{align}
\label{34}
    \begin{align}
    &{{\cal L}_{I_{_{sl}}^i}} \! \! \left( {s_s^i|\Phi \! \left( {{\cal A}_m^i} \right) \! =\!  0, \! S\! =\! i} \right) \! \! \nonumber \\
    & \mathop  = \limits^{\left( a \right)} \!\!\exp \!\! \left( \! \!{ \!- 2\pi \!\! \sum\limits_{k = 1}^K \!{{\lambda _S^k}\frac{{R_S^k}}{{{R_E}}} \!\! \int_{U_{UP}^1\left ( r \right )}^{U_{UP}^ 2\left ( r \right )} \!\!\!\! {\left(\!\! {1 \!-\! {{\left(\! {1 \!+\! s_s^i{p_k}{G_{s,k}}r_I^{ - \alpha } \! \beta } \right)}^{ \! - \chi }}} \! \right)\! {r_I}\! d{r_I}} } } \!\! \right) \nonumber\\
    & \times \!\! \exp \!\! \left( \!\! { \!- 2\pi \!\! \sum\limits_{k = 1}^K \!{{\lambda _S^k}\frac{{R_S^k}}{{{R_E}}} \!\!\int_{U_{LO}^1\left ( r \right )}^{U_{UP}^1\left ( r \right )} \!\!\!\! {\left( \!\! {1 \! - \! {{\left( \! {1 \! + \! s_s^i{p_k}{G_{m,k}}r_I^{ - \alpha } \! \beta \! } \right)}^{ \! - \chi }}} \! \right) \! {r_I}\! d{r_I}} } } \! \!  \right),
    \end{align}
where $(a)$ follows from the probability generating functional (PGFL) of the PPP \cite{haenggi2012stochastic}. 
\section{Proof of Theorem 1}
Each tier's CP comprises the CP when the serving satellite resides in the main-lobe and side-lobe regions, respectively. We get the CP of the whole system by multiplying the CP of each tier and the probability of the typical user connecting to each tier.
\label{35}
\begin{align}
&\mathbb{P}_C^E = \sum\limits_{i = 1}^K  \mathbb{P}\left\{ {{\rm{SINR}}_i^E \ge \gamma _{th}^i|S = i} \right\}\mathbb{P}_{ass,i}^E \nonumber\\
& = \sum\limits_{i = 1}^K {\left( \mathbb{P}{\left\{ {{\rm{SINR}}_i^E > \gamma _{th}^i|\Phi \left( {{\cal A}_m^i} \right) \!\!> \!\!0,\!S =\! i} \right\} \mathbb{P}_{m,i}^{E} \mathbb{P}_{ass,i}^E} \right.}  \nonumber \\
& + \left. { \mathbb{P}\left\{ {{\rm{SINR}}_i^E > \gamma _{th}^i|\Phi \left( {{\cal A}_m^i} \right) = 0,S = i} \right\} \mathbb{P}_{s,i}^{E} \mathbb{P}_{ass,i}^E} \right) \nonumber\\
& = \sum\limits_{i = 1}^K\! {\left(\!\!{\int_{{h_i}}^{r_m^i} \!\!\!\!\!\!\!\!{\Omega \left( {I_{ml}^i,s_m^i} \right) \!{f_{{R_i}}^{E}}\!\!\left( {r|\Phi \!\left( {{\cal A}_m^i} \right)\! >\! 0,S\! =\! i} \right)\!dr} \mathbb{P}_{m,ass,i}^{E}} \right.}  \nonumber\\
& + \!\!\!\!\left. {\int_{r_m^i}^{r_{\max }^i}\!\!\!\!\!\!\!\! {\Omega \left( {I_{sl}^i,s_s^i} \right){f_{R_i}^{E}}\left( {r|\Phi \left( {{\cal A}_m^i} \right) = 0,S = i} \right)dr} \mathbb{P}_{s,ass,i}^{E} } \right) \nonumber\\
& = \sum\limits_{i = 1}^K\! {\left(\!\! {\int_{{h_i}}^{r_m^i} \!\!\!\!\!\!\!\!\!{\Omega \!\left( {I_{ml}^i,s_m^i} \right)\!2\lambda _S^i\pi \frac{{R_S^i}}{{{R_E}}}r{e^{ - \!\sum\limits_{k = 1}^K {\!{\bf{1}}\left( {r \ge {h_k}} \right)\lambda _S^k\pi \frac{{R_S^k}}{{{R_E}}}\left( {{r^2} \!-\! h_k^2} \right)} }}\!dr} } \right.}  \nonumber \\
& + \!\!\!\left. {\int_{r_m^i}^{r_{\max }^i} \!\!\!\!\!\!\!\!\!\!\!\! {\Omega \left( {I_{sl}^i,s_s^i} \right)\!2\lambda _S^i\pi \frac{{R_S^i}}{{{R_E}}}r{e^{ - \!\sum\limits_{k = 1}^K \!{{\bf{1}}\left( {r \ge {h_k}} \right)\lambda _S^k \!\pi \frac{{R_S^k}}{{{R_E}}}\left( {{r^2} \!-\! h_k^2} \right)} }} \!dr} } \!\!\right). 
\end{align}
\vspace{-1em}
\section{Proof of Theorem 2}
\begin{figure*}[t]
\setlength{\abovecaptionskip}{-1.5cm}
\setlength{\belowcaptionskip}{-0.5cm}
\normalsize
\begin{equation}
\begin{aligned}
&{\bf{1}}\left( {t_E^i > t_{th}^i|S = i} \right) \mathop  = \limits^{\left( a \right)}  {\bf{1}}\left( {\cos \left( {\varphi _E^i - \arctan \left( {\tan \left( {\varphi _D^i} \right)\cos \left( \theta  \right)} \right)} \right) < \frac{{\cos \left( {\frac{{ct_{th}^i}}{{R_S^i}}} \right)}}{{\sqrt {1 - {{\left( {\sin \varphi _D^i} \right)}^2}{{\left( {\sin \left| \theta  \right|} \right)}^2}} }}} \right)\\
& = \left\{ \begin{array}{l}
1,{p_i}\left( \theta  \right) > 1\\
{\rm{     }}\arccos \left( {\frac{{{{\left( {R_S^i} \right)}^2} + R_E^2 - {r^2}}}{{2{R_E}R_S^i}}} \right) > \kappa \left( \theta  \right) + \arccos \left( {{p_i}\left( \theta  \right)} \right){\rm{  }},  {\rm{or  }}\arccos \left( {\frac{{{{\left( {R_S^i} \right)}^2} + R_E^2 - {r^2}}}{{2{R_E}R_S^i}}} \right) < \kappa \left( \theta  \right) - \arccos \left( {{p_i}\left( \theta  \right)} \right),\left| {{p_i}\left( \theta  \right)} \right| \le 1\\
0,{p_i}\left( \theta  \right) <  - 1
\end{array} \right.\\
&  \mathop  = \limits^{\left( b \right)}  \left\{ \begin{array}{l}
1,{p_i}\left( \theta  \right) > 1\\
R_{reg}^1 \cup R_{reg}^2,\left| {{p_i}\left( \theta  \right)} \right| \le 1\\
0,{p_i}\left( \theta  \right) <  - 1
\end{array} \right.
\label{36}
\end{aligned}
\end{equation}
${\vartheta^+}\!\!\left(  \theta  \right) \! = \!\!  {\bf{1}}\! \left( {\kappa \left( \theta  \right) \! + \! \arccos \left( {{p_i}\!\left( \theta  \right)} \right) \!>\! 0} \right) \!{\Xi _ + }\! \left( \theta  \right)$, 
${\vartheta ^ - }\left( \theta  \right) \! =\! {\bf{1}} \! \left( {\kappa \left( \theta  \right) \! -  \! \arccos \left( {{p_i}\left( \theta  \right)} \right) \!>\! 0} \right){\Xi _ - } \! \left( \theta  \right)$, \\
${\varsigma ^ + }\!\! \left( \theta  \right) \!$ $=$ $\! {\bf{1}} \! \left( {\kappa \left( \theta  \right) \! + \!  \arccos \left( {{p_i}\left( \theta  \right)} \right) \! <\!  0} \right)\!$ ${\bf{1}} \! \left(\! { - \left( {\kappa \left( \theta  \right) \! + \! \arccos \left( {{p_i}\left( \theta  \right)} \right)} \right) \! > \!  \arccos \!\left( \! {\frac{{{{\left( {R_S^i} \right)}^2} + R_E^2 - {r^2}}}{{2{R_E}R_S^i}}} \! \right)} \! \right){\Xi _ + }\! \left( \theta  \right)$, \\
${\Psi ^ + }\!\left( \theta  \right) \!$ $ = $ $\! {\bf{1}}\! \left( {\kappa \left( \theta  \right) \! + \! \arccos \left( {{p_i}\left( \theta  \right)} \right) \! < \! 0} \right)\! $ ${\bf{1}}\! \left( \!{\arccos \! \left(\! {\frac{{{{\left( {R_S^i} \right)}^2} + R_E^2 - {r^2}}}{{2{R_E}R_S^i}}} \! \right)\!  >\!  -\! \left( {\kappa \left( \theta  \right) \! + \! \arccos \left( {{p_i}\left( \theta  \right)} \right)} \right)} \!\right){\Xi _ + }\! \left( \theta  \right)$, \\
${\varsigma ^ - }\left( \theta  \right) $ $= $ ${\bf{1}}\left( {\kappa \left( \theta  \right) - \arccos \left( {{p_i}\left( \theta  \right)} \right)} >0\right)$ ${\bf{1}} \left( {\left| {\arccos \left( {\frac{{{{\left( {R_S^i} \right)}^2} + R_E^2 - {r^2}}}{{2{R_E}R_S^i}}} \right)} \right| < \kappa \left( \theta  \right) - \arccos \left( {{p_i}\left( \theta  \right)} \right)} \right){\Xi _ - }\left( \theta  \right)$, \\
${\Psi ^ - }\left( \theta  \right)$ $ = $ ${\bf{1}}\left( {\kappa \left( \theta  \right) - \arccos \left( {{p_i}\left( \theta  \right)} \right)} >0 \right) $ ${\bf{1}}\left( {\arccos \left( {\frac{{{{\left( {R_S^i} \right)}^2} + R_E^2 - {r^2}}}{{2{R_E}R_S^i}}} \right) <  - \left( {\kappa \left( \theta  \right) - \arccos \left( {{p_i}\left( \theta  \right)} \right)} \right)} \right){\Xi _ - }\left( \theta  \right)$,\\
${\Xi _ + }\!\left( \theta  \right)\! =\! \sqrt {{{\left( {R_S^i} \right)}^2} \!+\! R_E^2 \!-\! 2{R_E}R_S^i\cos \left( {\kappa \left( \theta  \right) \!+\! \arccos \left( {{p_i}\left( \theta  \right)} \right)} \right)} $, 
${\Xi _ - }\!\left( \theta  \right) \!=\! \sqrt {{{\left( {R_S^i} \right)}^2} \!+\! R_E^2 \!- \! 2{R_E}R_S^i\cos \left( {\kappa \left( \theta  \right) \!-\! \arccos \left( {{p_i}\left( \theta  \right)} \right)} \right)} $, 
$\kappa \left( \theta  \right) = \arctan \left( {\tan \left( {\varphi _D^i} \right)\cos \left( \theta  \right)} \right)$, 
${p_i}\left( \theta  \right) = \frac{{\cos \left( {\frac{{ct_{th}^i}}{{R_S^i}}} \right)}}{{\sqrt {1 - {{\left( {\sin \varphi _D^i} \right)}^2}{{\left( {\sin \left| \theta  \right|} \right)}^2}} }}$.

\hrulefill
\end{figure*} 
Using trigonometric functions, with $d_E^i \!=\! R_S^i $ $\arccos \! \left( \! {\sin {\varphi _E^i} \!\sin \varphi _D^i\cos \!\left( {\left| \theta  \right|} \right) \!+\! \cos {\varphi _E^i}\cos \varphi _D^i} \right)$,\cite{9511625,9918046,alzer1997some}, we obtain $(a)$ in (\ref{36}), where ${\varphi _E^i}$ represents the Earth-centered contact angle of $i$-th nearest satellite. By examining the trigonometric relationship between the nearest distance $r$ and ${\varphi _E^i}$ and conducting a detailed analysis of the upper and lower bounds of $r$, we convert the condition that the time for the nearest satellite to leave the dome region ${\varphi _D^i}$ exceeds the threshold time into a constraint on the distance from the nearest satellite to the typical user $r$.
For example, ${\bf{1}}\!\!\left( \!{\arccos \!\left( {{p_i}\left( \theta  \right)} \right) \!+\! \kappa \left( \theta  \right) \!<\! \arccos \! \left( \! {\frac{{{{\left( {R_S^i} \right)}^2} + R_E^2 - {r^2}}}{{2{R_E}R_S^i}}} \! \!\right)} \! \!\right) \!$ \!$=$ $\! \! {\left( {{\vartheta ^ + }\! \left( \theta  \right)\! < \! r} \right) \! \cup \! \left( {{\varsigma ^ + } \! \left( \theta  \right) \!>\! r} \right) \cup \left( {{\Psi ^ + } \!\left( \theta  \right) \! < \! r} \right)}  \! $ $=$ $\!  R_{reg}^1$, and $ {\bf{1}} \!\! \left(\! {\arccos\! \left( \! {\frac{{{{\left( {R_S^i} \right)}^2} + R_E^2 - {r^2}}}{{2{R_E}R_S^i}}} \!\right) \!\!  <\!  - \! \arccos \left( {{p_i}\left( \theta  \right)} \right) \! + \! \kappa \left( \theta \right)} \! \right) \! $\! $ = $ \!  $\! \left( {{\vartheta ^ - } \! \left( \theta  \right) \! < \! r} \right) \! $ $\cup $ $ \!  \left( {{\varsigma ^ - } \! \left( \theta  \right) \! > \! r} \right) \! $ $\cup \!  \left( {{\Psi ^ - } \! \left( \theta  \right) \! < \! r} \right) $ $= R_{reg}^2$. 
Combined with the conditional CP given in (\ref{17}), the NHP conditioned on accessing $i$-tier main-lobe region and side-lobe region are given in (\ref{19}) and (\ref{20}), respectively. 
$\Omega\left(I_{ml}^i, s_m^i\right)$ and $\Omega\left(I_{sl}^i, s_s^i\right)$ can be calculated by Theorem 1.
Multiplying the conditional NHPs given in (\ref{19}) and (\ref{20}) by the respective access probabilities for each tier and then summing them, we obtain the result presented in (\ref{18}).

\section{Proof of Proposition 3}
In a max-SINR satellite UAP system, we assume that the small-scale fading is correlated, thus neglecting its impact when calculating access probabilities. The probability of accessing tier $i$ is equivalent to the probability that the desired signal power from the nearest satellite in tier $i$ to the typical user exceeds that from all other tiers.
\label{37}
\begin{align}
&\mathbb{P}_{ass,i}^M= \int_{{h_i}}^{r_{max}^i} {\mathbb{P}{^M}\left\{ {S = i|{R_i} = r} \right\}{f_{{r}}^{E}}\left( r \right)dr} \nonumber \\
&=\!\! \int_{{h_i}}^{r_{max}^i}\!\!\!\! {\prod\limits_{k \ne i} \!\!{{e^{ \!- \!{\lambda _S^k}\pi \frac{{R_S^k}}{{{R_E}}}{\bf{1}}\left( \!{{{\left( {\frac{{{p_k}{G_{m,k}}}}{{{p_i}{G_{m,i}}}}} \right)}\!^{\frac{1}{\alpha }}}r > {h_k}} \!\right)\left(\! {{{\left( \!{\frac{{{p_k}{G_{m,k}}}}{{{p_i}{G_{m,i}}}}} \!\right)}\!^{\frac{2}{\alpha }}}{r^2} - h_k^2} \!\right)}}}\!\! {f_{{r}}^{E}}\!\left(\! r \right)\!\!dr} \nonumber \\
&= \!\!\!2{\lambda _S^i}\pi \frac{{R_S^i}}{{{R_E}}}\int_{{h_i}}^{r_{max}^i} {r{e^{ - \left( {{\psi _{i}^{MAX}}\left( r \right) + {\lambda _S^i}\pi \frac{{R_S^i}}{{{R_E}}}\left( {{r^2} \!-\! {h_i^2}} \right)} \right)}}\!dr},
\end{align}
where ${\psi_{i} ^{MAX}}\left( r \right)$ is given in Proposition 3, ${f_{{r}}^{E}}\!\left(\! r \right)$ represents the PDF of the distance distribution between the nearest satellite in tier $i$ and the typical user.

\section{Proof of Theorem 4}
Under the max-SINR satellite UAP, we divided the system's CP for the $K$-tier network into two parts, i.e., the satellite providing the maximum SINR is in the main-lobe region and the side-lobe, respectively. 
The overall CP is given below.
\begin{equation}
\label{38}
    \begin{aligned}
    &\mathbb{P}_C^M = \mathbb{P}\left\{ {\mathop {\max }\limits_{i,i \in K} \mathop {\max }\limits_{X \in \Phi _m^i} {\rm{SINR}}\left( X \right) > \gamma _{th}^i|\Phi \left( {{{\cal A}_m}} \right) > 0} \right\}\mathbb{P}_m^M\\
   &  + \mathbb{P}\left\{ {\mathop {\max }\limits_{i,i \in K} \mathop {\max }\limits_{X \in \Phi _s^i} {\rm{SINR}}\left( X \right) > \gamma _{th}^i|\Phi \left( {{{\cal A}_m}} \right) = 0} \right\}\mathbb{P}_s^M.
    \end{aligned}
\end{equation}
The CP, under the condition that the satellite providing the maximum SINR is in the main-lobe region, is given below.
\begin{equation}
\label{39}
\begin{aligned}
    &\mathbb{P}\left\{ {\mathop {\max }\limits_{i,i \in K} \mathop {\max }\limits_{X \in {\Phi _m^i}} {\rm{SINR}}\left( X \right) > {\gamma _{th}^i|\Phi \left( {{{\cal A}_m}} \right) > 0}} \right\}\\
    &= \mathbb{P}\left\{ {\mathop {\max }\limits_{i,i \in K} \frac{{{M_i^m}}}{{I^M - {M_i^m} + {\sigma^2}}} > {\gamma _{th}^i|\Phi \left( {{{\cal A}_m}} \right) > 0}} \right\}\\
    &\mathop  = \mathbb{P}\left\{ {\mathop {\max }\limits_{i } \left( {{M_i^m}{\delta _{th}^i}} \right) > I^{M} + {\sigma ^2}} \right\},
\end{aligned}
\end{equation}
where $M_{i}^m$ represents the maximum desired signal power in $i$-th tier, and $M_m^i = \mathop {\max }\limits_{X \in \Phi _m^i} {p_i}{G_{m,i}}{H_X}R_X^{ - \alpha }$, $\frac{{1 + {\gamma _{th}^i}}}{{{\gamma _{th}^i}}} = {\delta _{th}^i}$.
and $I^{M}$ is the total aggregated interference in the max-SINR UAP system, given by
${I^M} = \sum\limits_{k = 1}^K {\left( {\sum\limits_{X \in \Phi _m^k} {{p_k}{G_{m,k}}{H_X}R_X^{ - \alpha }}  + \sum\limits_{X \in \Phi _s^k} {{p_k}{G_{s,k}}{H_X}R_X^{ - \alpha }} } \right)} $.
The LTs of ${{I^M} + \sigma ^{2} }$ is given in (\ref{40}),
\begin{figure*}[t]
\setlength{\abovecaptionskip}{-1.5cm}
\setlength{\belowcaptionskip}{-0.5cm}
\normalsize
\begin{equation}
\begin{aligned}
&{\cal L}_{{I^M} + \sigma  ^{2} }^u\left( {s|\Phi \left( {{{\cal A}_m}} \right) > 0} \right) \mathop  = \limits^{\left( a \right)} \mathbb{E}\left\{ {{e^{ - s{\sigma ^{2}}}}} \right\} \mathbb{E} \left\{ {{e^{ - s{I^M}}} \times {\bf{1}}\left( {\mathop {\max }\limits_{i \in K} \left( {\delta _{th}^iM_i^m} \right) \le u} \right)} \right\}\\
&\mathop  = \limits^{\left( b \right)} \mathbb{E} \!\!\left\{ {{e^{ - s \sigma  ^{2}}}} \!\! \right\} \! \prod_{k=1}^{K} \!\! {e ^ {{ - \lambda _S^k2\pi \! \frac{{R_S^k}}{{{R_E}}} \!\! \int_{r_m^k}^{r_{\max }^k} \!\! {\left(\!  {1 - {E_{{H_X}}}\! \left\{ \! {{e^{ - s{p_k}{G_{s,k}} \! {H_X} \! r_I^{ - \alpha }\! }}} \! \right\}} \! \right){r_I}d{r_I}} } }} \!\! \times  \! \! 
\prod\limits_{k = 1}^K \!\! { e ^{{ \! - \lambda _S^k2\pi \! \frac{{R_S^k}}{{{R_E}}} \! \! \int_{{h_k}}^{r_m^k} \!\! {\left( \! {\! 1 - {E_{{H_X}}} \!\left\{ {{e^{ \!- s{p_k}{G_{m,k}}\! {H_X}r_I^{ - \alpha }}}\!{\bf{1}}\left(\! {{H_X} \le \frac{u r_I^{ \alpha }}{\delta _{th}^i{p_k}{G_{m,k}}} } \! \right)} \! \right\}} \! \right)\!{r_I}d{r_I}} }}  } \\
&\mathop  = \limits^{\left( c \right)} E\left\{ {{e^{ - j\omega \sigma  ^{2}}}} \right\}\exp \left( {\sum\limits_{k = 1}^K { - \lambda _S^k2\pi \frac{{R_S^k}}{{{R_E}}}\int_{r_m^k}^{r_{\max }^k} {\left( {1 - {{\left( {j\omega \beta {p_k}{G_{s,k}}r_I^{ - \alpha } + 1} \right)}^{ - \chi }}} \right){r_I}d{r_I}} } } \right) \times \\
&\exp \left( {\sum\limits_{k = 1}^K { - \lambda _S^k2\pi \frac{{R_S^k}}{{{R_E}}}\int_{{h_k}}^{r_m^k} {\left( {1 - \frac{{{{\left( {j\omega {p_k}{G_{m,k}}r_I^{ - \alpha } + {\beta ^{ - 1}}} \right)}^{ - \chi }}}}{{{\beta ^\chi }\Gamma \left( \chi  \right)}}\gamma \left( {\chi ,\frac{{j\omega y}}{{\delta _{th}^i}} + \frac{y}{{\beta \delta _{th}^i{p_k}{G_{m,k}}r_I^{ - \alpha }}}} \right)} \right){r_I}d{r_I}} } } \right).
\label{40}
\end{aligned}
\end{equation}
\hrulefill
\end{figure*} 
where $(a)$ comes from \cite{dhillon2011coverage}, $(b)$ is obtained from PGFL of PPP \cite{haenggi2012stochastic}, and $(c)$ is obtained from the LT of $H$ and \cite{gradshteyn2014table}.
Let $g  \left( u \right) \! = \!\! \sum\limits_{k = 1}^K \! - \!\lambda _S^k2\pi \frac{{R_S^k}}{{{R_E}}} \!\! \int_{{h_k}}^{r_m^k} \!\!  {\left( \! \! {1 \!\!  - \!\! \int_0^{\frac{u r_I^{ \alpha }}{{\delta _{th}^i{p_k}{G_{m,k}}}}}\!\!  {{e^{ - s{p_k}{G_{m,k}}hr_I^{ - \alpha }}} \! \! f\left( h \right) \! dh} } \!\! \right){r_I} d{r_I}} $. Differentiating $g  \left( u \right) $, we get 
\begin{align}
\label{41}
&\frac{{\partial g \! \left( u \right)}}{{\partial u}} \!\!= \! \sum\limits_{k = 1}^K \frac{\lambda _S^k2\pi {R_S^k}}{{{R_E}{\beta ^\chi }\Gamma \left( \chi  \right)}}{e^{ - \left( {\frac{{su}}{{\delta _{th}^i}}} \right)}}{{\left( {\delta _{th}^i{p_k}{G_{m,k}}} \right)}^{ - \chi }}{u^{\chi  - 1}} \times \nonumber  \\
& \int_{{h_k}}^{r_m^k} {r_I^{\alpha \chi  + 1}\exp \left( { - \frac{u}{{\beta \delta _{th}^i{p_k}{G_{m,k}}}}r_I^\alpha } \right)d{r_I}}  \nonumber \\
& \mathop  = \limits^{\left( a \right)} \sum\limits_{k = 1}^K  \frac{\lambda _S^k2\pi {R_S^k}}{{{R_E}\Gamma \left( \chi  \right) \alpha}}{e^{ - \left( {\frac{{j\omega y}}{{\delta _{th}^i}}} \right)}}{y^{ - \frac{2}{\alpha } - 1}}{{\left( {\beta \delta _{th}^i{p_k}{G_{m,k}}} \right)}^{\frac{2}{\alpha }}} \times \nonumber  \\
& \left( \! {\gamma \left( \!  {\frac{{\alpha \chi  + 2}}{\alpha },\frac{{y{{\left( {r_{\mathop{\rm m}\nolimits} ^k} \right)}^\alpha }}}{{\delta _{th}^i\beta {p_k}{G_{m,k}}}}} \right) \! -  \! \gamma \left( {\frac{{\alpha \chi  + 2}}{\alpha },\frac{{y{{\left( {{h_k}} \right)}^\alpha }}}{{\delta _{th}^i\beta {p_k}{G_{m,k}}}}} \!\right)} \! \right),
\end{align}
where $\left( a \right)$ is derived based on \cite{gradshteyn2014table}. Differentiating ${\cal L}_{I^{M} + {\sigma ^2}}^u\left( s \right)$, we get
\begin{equation}
\label{42}
    \begin{aligned} 
&\frac{{\partial {\cal L}_{{I^M } \!+ \! {\sigma ^2}}^u \!\! \left( {s|\Phi \!\! \left( {{{\cal A}_m}} \right) \! > 0} \right)}}{{\partial u}} = {\cal L}_{{I^M} + {\sigma ^2}}^u \!\! \left( \! {s|\Phi \left( {{{\cal A}_m}} \right) \! > \! 0} \right) \!\! \frac{{\partial g\left( u \right)}}{{\partial u}}\\
& = {\Theta ^M}\left( {s,u} \right){\Upsilon ^M}\left( {s,u} \right),
    \end{aligned}
\end{equation}
where ${\Theta ^M}\left( {s,u} \right) $ and $\Upsilon ^M\left( {s,u} \right)$ are given in (\ref{26}) and (\ref{27}). 
Let $f\left( {x,y} \right)$ denotes the joint density of $I^{M} + {\sigma ^2}$ and $\max \left( {{\delta _{th}^i}{M_i^m}} \right)$. The CP equals 
\begin{align}
\label{43}
& \mathop  = \mathbb{P}\left\{ {\mathop {\max }\limits_{i } \left( {{M_i^m}{\delta _{th}^i}} \right) > I^{M} + {\sigma^2}} \right\} \nonumber \\
&= \int_0^\infty  {\int_0^\infty  {f\left( {x,y} \right){\bf{1}}\left( {y > x > \frac{y}{{\max \left( {{\delta _{th}^i}} \right)}}} \right)} dxdy} \nonumber  \\
&\mathop  = \limits^{\left( a \right)} \int_0^\infty  {\int_0^\infty  {f\left( {x,y} \right)\int_{\beta  - j\infty }^{\beta  + j\infty } {\frac{{{e^{ - s\frac{y}{{\max \left( {{\delta _{th}^i}} \right)}}}} - {e^{ - sy}}}}{{2\pi js}}{e^{sx}}ds} dxdy} } \nonumber \\
&\mathop  = \limits^{\left( b \right)} \int_0^\infty  {\int_0^\infty  {f\left( {x,y} \right){e^{ - sx}}dx\int_{\beta  - j\infty }^{\beta  + j\infty } {\frac{{{e^{sy}} - {e^{s\frac{y}{{\max \left( {{\delta _{th}^i}} \right)}}}}}}{{2\pi js}}dsdy} } } \nonumber \\
&\mathop  = \limits^{\left( c \right)} \int_0^\infty  {\int_{\beta  - j\infty }^{\beta  + j\infty } {\frac{{{e^{sy}} - {e^{s\frac{y}{{\max \left( {{\delta _{th}^i}} \right)}}}}}}{{2\pi js}}\widehat f\left( {s,y} \right)dsdy} } \nonumber \\
&\mathop  = \limits^{\left( d \right)} \int_0^\infty  {\int_{ - \infty }^{ + \infty } {\frac{{{e^{j\omega y}} - {e^{j\omega \frac{y}{{\max \left( {{\delta _{th}^i}} \right)}}}}}}{{2\pi j\omega }}\widehat f\left( {j\omega ,y} \right)d\omega dy} }, 
\end{align}
where $$\widehat f\left( {j\omega ,y} \right) = \frac{{\partial {\cal L}_{I^{M} + {\sigma ^2}}^u\left( {j\omega } \right)}}{{\partial u}} = \int_{x = 0}^\infty  {f\left( {x,y} \right){e^{ - j\omega x}}dx}, $$ 
$(a)$ arises from the laplace transform of $x$ in ${\bf{1}}\left( {y > x > \frac{y}{{\max \left( {{\delta _i}} \right)}}} \right)$, $(b)$ employs a substitution method, $(c)$ involves a bivariate laplace transform, where ${\widehat f\left( {s,y} \right)}$ represents the Laplace transform of $f\left( {x,y} \right)$ with respect to $x$, and $(d)$ represents the transformation between laplace transform and fourier transform when ${\mathop{\rm Re}\nolimits} \left\{ s \right\} = 0$.

By substituting the derivation result from (\ref{42}) into (\ref{43}), and calculate $ \mathbb{P}\left\{ {\mathop {\max }\limits_{i,i \in K} \mathop {\max }\limits_{X \in \Phi _s^i} {\rm{SINR}}\left( X \right) > \gamma _{th}^i|\Phi \left( {{{\cal A}_m}} \right) = 0} \right\}$ with the same method, we obtain the outcome presented in (\ref{25}).
\end{appendices}



\ifCLASSOPTIONcaptionsoff
  \newpage
\fi

\vspace{-2mm}
\bibliographystyle{IEEEtran}
\bibliography{ref}

\begin{thebibliography}{10}
\providecommand{\url}[1]{#1}
\csname url@samestyle\endcsname
\providecommand{\newblock}{\relax}
\providecommand{\bibinfo}[2]{#2}
\providecommand{\BIBentrySTDinterwordspacing}{\spaceskip=0pt\relax}
\providecommand{\BIBentryALTinterwordstretchfactor}{4}
\providecommand{\BIBentryALTinterwordspacing}{\spaceskip=\fontdimen2\font plus
\BIBentryALTinterwordstretchfactor\fontdimen3\font minus \fontdimen4\font\relax}
\providecommand{\BIBforeignlanguage}[2]{{%
\expandafter\ifx\csname l@#1\endcsname\relax
\typeout{** WARNING: IEEEtran.bst: No hyphenation pattern has been}%
\typeout{** loaded for the language `#1'. Using the pattern for}%
\typeout{** the default language instead.}%
\else
\language=\csname l@#1\endcsname
\fi
#2}}
\providecommand{\BIBdecl}{\relax}
\BIBdecl

\bibitem{ohlen2016data}
P.~{\"O}hl{\'e}n, B.~Skubic, A.~Rostami, M.~Fiorani, P.~Monti, Z.~Ghebretensa{\'e}, J.~M{\aa}rtensson, K.~Wang, and L.~Wosinska, ``Data plane and control architectures for 5g transport networks,'' \emph{Journal of Lightwave Technology}, vol.~34, no.~6, pp. 1501--1508, 2016.

\bibitem{nguyen20216g}
D.~C. Nguyen, M.~Ding, P.~N. Pathirana, A.~Seneviratne, J.~Li, D.~Niyato, O.~Dobre, and H.~V. Poor, ``6g internet of things: A comprehensive survey,'' \emph{IEEE Internet of Things Journal}, vol.~9, no.~1, pp. 359--383, 2021.

\bibitem{kawamoto2013traffic}
Y.~Kawamoto, H.~Nishiyama, N.~Kato, and N.~Kadowaki, ``A traffic distribution technique to minimize packet delivery delay in multilayered satellite networks,'' \emph{IEEE Transactions on Vehicular Technology}, vol.~62, no.~7, pp. 3315--3324, 2013.

\bibitem{9520380}
B.~Shang, Y.~Yi, and L.~Liu, ``Computing over space-air-ground integrated networks: Challenges and opportunities,'' \emph{IEEE Network}, vol.~35, no.~4, pp. 302--309, 2021.

\bibitem{wei2021hybrid}
T.~Wei, W.~Feng, Y.~Chen, C.-X. Wang, N.~Ge, and J.~Lu, ``Hybrid satellite-terrestrial communication networks for the maritime internet of things: Key technologies, opportunities, and challenges,'' \emph{IEEE Internet of things journal}, vol.~8, no.~11, pp. 8910--8934, 2021.

\bibitem{li2020maritime}
X.~Li, W.~Feng, Y.~Chen, C.-X. Wang, and N.~Ge, ``Maritime coverage enhancement using uavs coordinated with hybrid satellite-terrestrial networks,'' \emph{IEEE Transactions on Communications}, vol.~68, no.~4, pp. 2355--2369, 2020.

\bibitem{xu2023space}
J.~Xu, M.~A. Kishk, and M.-S. Alouini, ``Space-air-ground-sea integrated networks: Modeling and coverage analysis,'' \emph{IEEE Transactions on Wireless Communications}, 2023.

\bibitem{10685064}
J.~Li, L.~Yang, Q.~Wu, X.~Lei, F.~Zhou, F.~Shu, X.~Mu, Y.~Liu, and P.~Fan, ``Active ris-aided noma-enabled space- air-ground integrated networks with cognitive radio,'' \emph{IEEE Journal on Selected Areas in Communications}, vol.~43, no.~1, pp. 314--333, 2025.

\bibitem{del2019technical}
I.~Del~Portillo, B.~G. Cameron, and E.~F. Crawley, ``A technical comparison of three low earth orbit satellite constellation systems to provide global broadband,'' \emph{Acta astronautica}, vol. 159, pp. 123--135, 2019.

\bibitem{mcdowell2020low}
J.~C. McDowell, ``The low earth orbit satellite population and impacts of the spacex starlink constellation,'' \emph{The Astrophysical Journal Letters}, vol. 892, no.~2, p. L36, 2020.

\bibitem{10689625}
Z.~Li and B.~Shang, ``Fundamentals of satellite-maritime communications: Downlink and uplink analysis,'' \emph{IEEE Transactions on Communications}, vol.~73, no.~4, pp. 2191--2206, 2025.

\bibitem{okati2020downlink}
N.~Okati, T.~Riihonen, D.~Korpi, I.~Angervuori, and R.~Wichman, ``Downlink coverage and rate analysis of low earth orbit satellite constellations using stochastic geometry,'' \emph{IEEE Transactions on Communications}, vol.~68, no.~8, pp. 5120--5134, 2020.

\bibitem{al2021analytic}
A.~Al-Hourani, ``An analytic approach for modeling the coverage performance of dense satellite networks,'' \emph{IEEE Wireless Communications Letters}, vol.~10, no.~4, pp. 897--901, 2021.

\bibitem{10387244}
B.~Shang, X.~Li, C.~Li, and Z.~Li, ``Coverage in cooperative leo satellite networks,'' \emph{Journal of Communications and Information Networks}, vol.~8, no.~4, pp. 329--340, 2023.

\bibitem{10430115}
L.~Yang, J.~Xiang, S.~Li, X.~Li, K.~Guo, M.~O. Hasna, and P.~S. Bithas, ``Performance analysis of relay-aided satellite-underwater acoustic communication systems,'' \emph{IEEE Transactions on Communications}, vol.~72, no.~6, pp. 3511--3525, 2024.

\bibitem{guo2024user}
Y.~Guo, C.~Skouroumounis, S.~Chatzinotas, and I.~Krikidis, ``On user association in large-scale heterogeneous leo satellite network,'' \emph{IEEE Transactions on Aerospace and Electronic Systems}, 2024.

\bibitem{andrews2016primer}
J.~G. Andrews, A.~K. Gupta, and H.~S. Dhillon, ``A primer on cellular network analysis using stochastic geometry,'' \emph{arXiv preprint arXiv:1604.03183}, 2016.

\bibitem{dhillon2012modeling}
H.~S. Dhillon, R.~K. Ganti, F.~Baccelli, and J.~G. Andrews, ``Modeling and analysis of k-tier downlink heterogeneous cellular networks,'' \emph{IEEE Journal on Selected Areas in Communications}, vol.~30, no.~3, pp. 550--560, 2012.

\bibitem{dhillon2011coverage}
------, ``Coverage and ergodic rate in k-tier downlink heterogeneous cellular networks,'' in \emph{2011 49th Annual Allerton Conference on Communication, Control, and Computing (Allerton)}.\hskip 1em plus 0.5em minus 0.4em\relax IEEE, 2011, pp. 1627--1632.

\bibitem{madhusudhanan2016analysis}
P.~Madhusudhanan, J.~G. Restrepo, Y.~Liu, and T.~X. Brown, ``Analysis of downlink connectivity models in a heterogeneous cellular network via stochastic geometry,'' \emph{IEEE Transactions on Wireless Communications}, vol.~15, no.~6, pp. 3895--3907, 2016.

\bibitem{10530195}
B.~Shang, X.~Li, Z.~Li, J.~Ma, X.~Chu, and P.~Fan, ``Multi-connectivity between terrestrial and non-terrestrial mimo systems,'' \emph{IEEE Open Journal of the Communications Society}, vol.~5, pp. 3245--3262, 2024.

\bibitem{9777886}
B.~Shang, E.~S. Bentley, and L.~Liu, ``Uav swarm-enabled aerial reconfigurable intelligent surface: Modeling, analysis, and optimization,'' \emph{IEEE Transactions on Communications}, vol.~71, no.~6, pp. 3621--3636, 2023.

\bibitem{9178984}
B.~Shang, L.~Liu, H.~Chen, C.~J. Zhang, S.~Pudlewski, E.~S. Bentley, and J.~D. Ashdown, ``Spatial spectrum sensing in uplink two-tier user-centric deployed hetnets,'' \emph{IEEE Transactions on Wireless Communications}, vol.~19, no.~12, pp. 7957--7972, 2020.

\bibitem{park2022tractable}
J.~Park, J.~Choi, and N.~Lee, ``A tractable approach to coverage analysis in downlink satellite networks,'' \emph{IEEE Transactions on Wireless Communications}, vol.~22, no.~2, pp. 793--807, 2022.

\bibitem{okati2022nonhomogeneous}
N.~Okati and T.~Riihonen, ``Nonhomogeneous stochastic geometry analysis of massive leo communication constellations,'' \emph{IEEE Transactions on Communications}, vol.~70, no.~3, pp. 1848--1860, 2022.

\bibitem{jung2022performance}
D.-H. Jung, J.-G. Ryu, W.-J. Byun, and J.~Choi, ``Performance analysis of satellite communication system under the shadowed-rician fading: A stochastic geometry approach,'' \emph{IEEE Transactions on Communications}, vol.~70, no.~4, pp. 2707--2721, 2022.

\bibitem{wang2021multi}
P.~Wang, B.~Di, and L.~Song, ``Multi-layer leo satellite constellation design for seamless global coverage,'' in \emph{2021 IEEE Global Communications Conference (GLOBECOM)}.\hskip 1em plus 0.5em minus 0.4em\relax IEEE, 2021, pp. 01--06.

\bibitem{choi2024cox}
C.-S. Choi \emph{et~al.}, ``Cox point processes for multi altitude leo satellite networks,'' \emph{IEEE Transactions on Vehicular Technology}, 2024.

\bibitem{talgat2020stochastic}
A.~Talgat, M.~A. Kishk, and M.-S. Alouini, ``Stochastic geometry-based analysis of leo satellite communication systems,'' \emph{IEEE Communications Letters}, vol.~25, no.~8, pp. 2458--2462, 2020.

\bibitem{qiu2023interference}
Z.~Qiu, W.~Wang, J.~Geng, and Y.~Liu, ``Interference analysis of multi-tier ngso based on stochastic geometry,'' in \emph{2023 IEEE Wireless Communications and Networking Conference (WCNC)}.\hskip 1em plus 0.5em minus 0.4em\relax IEEE, 2023, pp. 1--6.

\bibitem{yim2024modeling}
J.~Yim, J.~Park, and N.~Lee, ``Modeling and coverage analysis of k-tier integrated satellite-terrestrial downlink networks,'' \emph{arXiv preprint arXiv:2403.11096}, 2024.

\bibitem{9678973}
D.-H. Jung, J.-G. Ryu, W.-J. Byun, and J.~Choi, ``Performance analysis of satellite communication system under the shadowed-rician fading: A stochastic geometry approach,'' \emph{IEEE Transactions on Communications}, vol.~70, no.~16, pp. 2707--2721, 2022.

\bibitem{alkhateeb2017initial}
A.~Alkhateeb, Y.-H. Nam, M.~S. Rahman, J.~Zhang, and R.~W. Heath, ``Initial beam association in millimeter wave cellular systems: Analysis and design insights,'' \emph{IEEE Transactions on Wireless Communications}, vol.~16, no.~5, pp. 2807--2821, 2017.

\bibitem{zhu2018secrecy}
Y.~Zhu, G.~Zheng, and M.~Fitch, ``Secrecy rate analysis of uav-enabled mmwave networks using mat{\'e}rn hardcore point processes,'' \emph{IEEE Journal on Selected Areas in Communications}, vol.~36, no.~7, pp. 1397--1409, 2018.

\bibitem{dabiri20203d}
M.~T. Dabiri, M.~Rezaee, V.~Yazdanian, B.~Maham, W.~Saad, and C.~S. Hong, ``3d channel characterization and performance analysis of uav-assisted millimeter wave links,'' \emph{IEEE Transactions on Wireless Communications}, vol.~20, no.~1, pp. 110--125, 2020.

\bibitem{chen2023coverage}
Q.~Chen, W.~Meng, S.~Han, C.~Li, and T.~Q. Quek, ``Coverage analysis of sagin with sectorized beam pattern under shadowed-rician fading channels,'' \emph{IEEE Transactions on Communications}, vol.~71, no.~8, pp. 4988--5004, 2023.

\bibitem{8894851}
Q.~Huang, M.~Lin, W.-P. Zhu, S.~Chatzinotas, and M.-S. Alouini, ``Performance analysis of integrated satellite-terrestrial multiantenna relay networks with multiuser scheduling,'' \emph{IEEE Transactions on Aerospace and Electronic Systems}, vol.~56, no.~4, pp. 2718--2731, 2020.

\bibitem{9497773}
D.-H. Na, K.-H. Park, Y.-C. Ko, and M.-S. Alouini, ``Performance analysis of satellite communication systems with randomly located ground users,'' \emph{IEEE Transactions on Wireless Communications}, vol.~21, no.~1, pp. 621--634, 2022.

\bibitem{1198102}
A.~Abdi, W.~Lau, M.-S. Alouini, and M.~Kaveh, ``A new simple model for land mobile satellite channels: first- and second-order statistics,'' \emph{IEEE Transactions on Wireless Communications}, vol.~2, no.~3, pp. 519--528, 2003.

\bibitem{9511625}
H.~Jia, Z.~Ni, C.~Jiang, L.~Kuang, and J.~Lu, ``Uplink interference and performance analysis for megasatellite constellation,'' \emph{IEEE Internet of Things Journal}, vol.~9, no.~6, pp. 4318--4329, 2022.

\bibitem{9918046}
H.~Jia, C.~Jiang, L.~Kuang, and J.~Lu, ``An analytic approach for modeling uplink performance of mega constellations,'' \emph{IEEE Transactions on Vehicular Technology}, vol.~72, no.~2, pp. 2258--2268, 2023.

\bibitem{al2021session}
A.~Al-Hourani, ``Session duration between handovers in dense leo satellite networks,'' \emph{IEEE Wireless Communications Letters}, vol.~10, no.~12, pp. 2810--2814, 2021.

\bibitem{gradshteyn2014table}
I.~S. Gradshteyn and I.~M. Ryzhik, \emph{Table of integrals, series, and products}.\hskip 1em plus 0.5em minus 0.4em\relax Academic press, 2014.

\bibitem{wang2023coverage}
X.~Wang, N.~Deng, and H.~Wei, ``Coverage and rate analysis of leo satellite-to-airplane communication networks in terahertz band,'' \emph{IEEE Transactions on Wireless Communications}, vol.~22, no.~12, pp. 9076--9090, 2023.

\bibitem{okati2020stochastic}
N.~Okati and T.~Riihonen, ``Stochastic analysis of satellite broadband by mega-constellations with inclined leos,'' in \emph{2020 IEEE 31st Annual International Symposium on Personal, Indoor and Mobile Radio Communications}.\hskip 1em plus 0.5em minus 0.4em\relax IEEE, 2020, pp. 1--6.

\bibitem{song2022cooperative}
Z.~Song, J.~An, G.~Pan, S.~Wang, H.~Zhang, Y.~Chen, and M.-S. Alouini, ``Cooperative satellite-aerial-terrestrial systems: A stochastic geometry model,'' \emph{IEEE Transactions on Wireless Communications}, vol.~22, no.~1, pp. 220--236, 2022.

\bibitem{haenggi2012stochastic}
M.~Haenggi, \emph{Stochastic geometry for wireless networks}.\hskip 1em plus 0.5em minus 0.4em\relax Cambridge University Press, 2012.

\bibitem{alzer1997some}
H.~Alzer, ``On some inequalities for the incomplete gamma function,'' \emph{Mathematics of Computation}, vol.~66, no. 218, pp. 771--778, 1997.

\end{thebibliography}

\end{document}